\documentclass[12pt,letterpaper]{article}
\usepackage{amsmath,amssymb,amsfonts,amsthm,dsfont,mathtools,lipsum,etoolbox,mathabx}
\usepackage{booktabs}
\usepackage[flushleft]{threeparttable}
\usepackage{multirow}
\usepackage{setspace}
\usepackage{color}
\usepackage{endnotes}
\usepackage{enumerate}
\usepackage{float}
\usepackage{makecell}
\usepackage{afterpage}
\usepackage[bottom]{footmisc}
\usepackage[pdftex]{graphicx} 
\usepackage{epstopdf}
\usepackage{siunitx}
\usepackage{tikz}
\usepackage{indentfirst}
\usepackage{latexsym}
\usepackage{lscape}
\usepackage{parskip}
\usepackage{rotating}
\usepackage{caption}
\usepackage[normalem]{ulem}
\usepackage{pgfplots}
\usepackage{bm}
\usepackage{verbatim}
\usepackage{mathdots}
\usepackage{soul}
\usepackage{mwe}
\usepackage{arydshln}
\usepackage{xcolor}
\usepackage{subcaption}
\usepackage[scaled=0.85]{beramono}
\usepackage{titlesec}
\usepackage{placeins}
\usepackage[colorlinks=true,allcolors=blue]{hyperref}
\usepackage{cleveref}
\usepackage{chngcntr}
\usepackage{scalerel,stackengine}
\usepackage{anyfontsize}
\usepackage[left=1.0in,right=1.0in,top=1.0in,bottom=1.0in]{geometry}
\usepackage[shrink=0,letterspace=0]{microtype}
\counterwithin{figure}{section}
\counterwithin{table}{section} 

\usepackage[authoryear, round,semicolon,sort&compress]{natbib}
\usepackage{mathptmx}
\usepackage[T1]{fontenc}

\newtheorem{asn}{Assumption}[section]

\newtheorem{prop}{Proposition}[section]

\definecolor{mygreen}{rgb}{0.2, 0.7, 0.1}
\definecolor{lightblue}{rgb}{0.7,0.95,1}   
\definecolor{lightorange}{rgb}{1,0.9,0.7}  

\makeatletter
\newcommand{\leqnomode}{\tagsleft@true}
\newcommand{\reqnomode}{\tagsleft@false}
\makeatother

\newcommand\redsout{\bgroup\markoverwith{\textcolor{red}{\rule[0.5ex]{1pt}{1.4pt}}}\ULon}
\newcommand\bluesout{\bgroup\markoverwith{\textcolor{blue}{\rule[0.5ex]{1pt}{1.4pt}}}\ULon}

\newcommand{\supp}{\textup{Support}}

\newcommand{\indep}{\mathrel{\perp\!\!\!\perp}}
\newcommand{\Supp}{\text{Support}}
\newcommand{\eps}{\varepsilon}

\makeatletter
\newcommand{\VBig}{\bBigg@{2.7}}
\newcommand{\Vast}{\bBigg@{3.6}}
\newcommand{\VVast}{\bBigg@{5.5}}

\stackMath
\newcommand\reallywidehat[1]{%
	\savestack{\tmpbox}{\stretchto{%
			\scaleto{%
				\scalerel*[\widthof{\ensuremath{#1}}]{\kern-.6pt\bigwedge\kern-.6pt}%
				{\rule[-\textheight/2]{1ex}{\textheight}}
			}{\textheight}%
		}{0.5ex}}%
	\stackon[1pt]{#1}{\tmpbox}%
}

\makeatother 

\allowdisplaybreaks[4]

\newcommand{\zerodisplayskips}{%
	\setlength{\abovedisplayskip}{0.15cm}
	\setlength{\belowdisplayskip}{0.15cm}
	\setlength{\abovedisplayshortskip}{0.15cm}
	\setlength{\belowdisplayshortskip}{0.15cm}}
\appto{\normalsize}{\zerodisplayskips}
\appto{\small}{\zerodisplayskips}
\appto{\footnotesize}{\zerodisplayskips}

\titleformat{\section}{\fontsize{14}{15.6}\bfseries\selectfont}{\thesection}{1em}{}
\titleformat{\subsection}{\fontsize{12.5}{14.4}\bfseries\selectfont}{\thesubsection}{1em}{}
\titleformat{\subsubsection}{\fontsize{12.5}{14.4}\bfseries\itshape\selectfont}{\thesubsubsection}{1em}{}

\titlespacing\section{0pt}{8pt}{3pt}
\titlespacing\subsection{0pt}{6pt}{2pt}
\titlespacing\subsubsection{0pt}{4pt}{2pt}

\begin{document}	
	
  	  	\title{Assignment at the Frontier: Identifying the Frontier Structural Function and Bounding Mean Deviations\thanks{
  	  	We thank Xavier D'haultfoeuille, Jinyong Hahn, 
  	  	conference participants at AEFI 2025, and seminar participants at Ben-Gurion University, Hebrew University of Jerusalem,  Penn State, Oxford, and Aarhus 
  	  	for helpful comments and discussions.}}
    	%
			
\author{Dan Ben-Moshe\thanks{Department of Economics, Ben-Gurion University of the Negev. Email: dbmster@gmail.com}  
	\and David Genesove\thanks{Department of Economics, The Hebrew University of Jerusalem. Email: david.genesove@mail.huji.ac.il}}
	
\date{\today}
  
   \maketitle
  \begin{abstract}	
  	
This paper analyzes a model in which an outcome equals a frontier function of inputs minus a nonnegative unobserved deviation. The inputs may be endogenous (statistically dependent on the deviation). If zero lies in the support of the deviation given the inputs---an assumption we term assignment at the frontier---then the frontier is identified by the supremum of the outcome given those inputs, obviating the need for instruments. We then consider estimation with random error that is mean-independent of the inputs. Motivated by the assignment at the frontier assumption, we regularize estimation by requiring the fitted distribution of the deviation to maintain a minimum probability mass in a neighborhood of zero. Finally, we derive a lower bound on mean deviation, using only variance and skewness, that is robust to scarcity of data near the frontier. We apply our methods to estimate a frontier production function and mean inefficiency. 

\vspace{0.4cm}

\noindent \textbf{Keywords:} Assignment at the frontier, frontier structural function, deviations, stochastic frontier analysis, production function, identification, bound, regularization. 


  \end{abstract}

\newpage

\section{Introduction}


This paper analyzes a model in which an outcome equals a function of inputs minus a nonnegative unobserved deviation. The function represents a frontier structural function (FSF) relating inputs to the outcome; the nonnegative unobservable is the deviation from the FSF.\footnote{The FSF is analogous to the average structural function (ASF) defined by \cite{blundell2003endogeneity}.} The assumption that deviations are nonnegative is motivated by applications where they represent inefficiencies or other wedges, such as regulations, markups, or taxes. The inputs may be endogenous (statistically dependent on the deviation). If zero lies in the support of the deviation given inputs---an assumption we term \textit{assignment at the frontier}---then the FSF is identified by the supremum of the outcome at those inputs, even when inputs are correlated with the deviation. Like random assignment, assignment at the frontier obviates the need for instrumental variables; unlike random assignment, inputs need not be exogenous.

For example, consider a production function where output equals a linear function of labor minus (nonnegative) inefficiency. When firms choose labor optimally based on productivity, labor is endogenous, and hence correlated with inefficiency. As a result, the conditional mean of output given labor does not identify the marginal productivity of labor, because as labor varies, changes in mean output conflate the effect of labor with changes in average inefficiency (so ordinary least squares is biased). However, when inefficiency enters separably,\footnote{Section~\ref{se:ASF} discusses when the frontier identifies the production and mean production functions under additive separability and input rescaling, and when input-inefficiency interactions prevent identification.} the conditional supremum of output given labor does identify the marginal productivity: as labor varies, changes in frontier output are due only to the effect of labor, since frontier firms have zero inefficiency given labor. The source of identification differs. Identification of the mean production function typically comes from mean independence of productivity and a set of exogenous variables. In contrast, identification of the frontier production function comes from the presence of firms arbitrarily close to efficiency, given inputs.

The identifying assumption that zero lies in the support of the deviation given inputs is equivalent to the infimum of the deviation equaling zero. Assignment at the frontier imposes no further restriction on the joint distribution of the deviation and inputs. In particular, the assumption allows the density of the deviation to vary with inputs at all points of its support. Figure \ref{fig:distu} illustrates this, showing how the FSF corresponds to the maximum of the outcome’s support at each input value (that is, the upper envelope of the outcome distribution), and that the distribution of the deviation varies with the input. Thus, inputs need not be exogenous, and neither the mean nor any quantile is necessarily a constant downward shift from the FSF.

\begin{figure}[!t]
	\centering
	
	\def\leftskewdist#1#2#3{
		\draw[domain=0:3, smooth, variable=\y, black] 
		plot ({#1 + #3 * ((3-\y)^1.5) * exp(-(3-\y)^1.5)}, {#2 + \y});
	}
	
	\def\rightskewdist#1#2#3{
		\draw[domain=-0.1:3.5, smooth, variable=\y, black] 
		plot ({#1 + #3 * (\y^1.5) * exp(-\y^1.5)}, {#2 + \y});
	}
	
	\def\normaldist#1#2#3{
		\draw[domain=0:3.9, smooth, variable=\y, black] 
		plot ({#1 + #3 * exp(-0.5 * (\y - 1.95)^2)}, {#2 + \y});
	}
	
	\definecolor{myblue}{RGB}{55, 126, 184}
	\definecolor{myorange}{RGB}{255, 127, 0}
	
	\resizebox{0.78\textwidth}{!}{
	\begin{tikzpicture}[
		xscale=1.1,  
		yscale=0.9,  
		font=\footnotesize
		]
		
		\draw[->] (0,1.5) -- (8,1.5);
		\node[anchor=north] at (4,1.495) {Input ($x$)};
		
		\draw[->] (0,1.5) -- (0,7);
		\node[rotate=90, anchor=south, xshift=-2pt] at (0,4.25) {Outcome ($y$)} ;
		
		\draw[thick, black, domain=0:7, smooth] plot ({\x}, {-0.05*(\x - 6)^2 + 6.5});
		\node at (4.0, 7.2) {\footnotesize FSF};
		\draw[->, thick] (4.1, 7) -- (3.6, 6.3);
		
		\draw[thick, dashed, myblue, domain=0:6.82, smooth] plot ({\x}, {0.3*sin(\x r) + 4.0});
		\node[myblue] at (5.2, 4.7)  {\footnotesize Mean};
		\draw[->, thick, myblue] (5.2, 4.5) -- (4.9, 3.8);
		
		\draw[thick, dotted, myorange, domain=0:6.25, smooth] plot ({\x}, {0.1*sin(\x r) + 3.0});
		\node[myorange] at (5.1, 2.0) {\footnotesize Example quantile};
		\draw[->, thick, myorange] (5.1, 2.3) -- (4.9, 2.8);
		
		\leftskewdist{1.8}{2.6}{1.2};  
		\rightskewdist{3.9}{2.8}{1.5};  
		\normaldist{6}{2.6}{0.9};  
		
		\node at (9.2, 5.6) [align=center] {\footnotesize Density of FSF \\ \footnotesize minus deviation};
		\draw[->, thick] (8.0, 5.6) -- (7.0, 5.0);
		
	\end{tikzpicture}
}
	 \caption{	 The FSF is the maximum of the outcome’s support at each input. 
	 	The distribution of the deviation varies with inputs.}
	\label{fig:distu}
\end{figure}

As Figure \ref{fig:distu} suggests, the FSF is identified by the supremum of the outcome given inputs, making a nonparametric Data Envelopment Analysis (DEA) frontier a natural estimator \citep{charnes1978measuring,farrell1957measurement}. DEA, however, ignores random errors. In our setting, we observe inputs and outcomes; we do not observe any auxiliary indicator or proxy for proximity to the frontier. As a result, a sample supremum of the outcome given inputs may reflect a large positive error rather than a near-frontier observation. We therefore model outcomes using stochastic frontier analysis (SFA) to separate deviations from random error. Unlike much of the SFA literature, we allow the deviation to be statistically dependent on inputs. Without such dependence, the frontier reduces to a constant upward shift of the conditional mean, implicitly imposing exogeneity. The random error is assumed mean-independent of inputs.

In a cross-sectional setting, identification strategies rely on support restrictions (deviations on the nonnegative reals and errors on the reals) and often impose symmetry on the random error. Recent work shows that when both the frontier and deviations depend on the same observables, finite-sample estimation can have difficulty separating the two \citep{parmeter2024inference}. We therefore focus on a panel setting with time-invariant deviations, which exploits within- and between-firm variation for identification rather than symmetry or support restrictions on the random error.

We implement this strategy using a method of moments estimator \citep{olson1980monte} adapted to the panel setting. First, the outcome is nonparametrically regressed on inputs to estimate the conditional mean outcome. 
	 This regression includes any endogeneity bias arising from  dependence between inputs and the deviation. Second, we use within–between decompositions to form bias-corrected powers of the residuals and, following approaches for variance estimation under heteroskedasticity \citep{hall1989variance, fan1998efficient}, flexibly regress these bias-corrected powers of the residuals on inputs to estimate conditional central moments of the deviation. Third, at each input value we specify a parametric conditional distribution for the deviation and estimate its parameters by method of moments, subject to a near-frontier mass constraint (a regularization that enforces a minimum probability mass near the frontier). Finally, we compute the conditional mean deviation implied by these estimated parameters and add it to the estimated conditional mean outcome to obtain an estimate of the frontier at a given input value. This estimation strategy allows the conditional distribution of the deviation to vary flexibly with inputs.\footnote{For large-sample properties and inference of related estimators, see \citet{simar2017nonparametric} and \citet{parmeter2024inference}.} 

The regularization imposed by the near-frontier mass constraint addresses the fact that estimated central moments may be only weakly informative about near-frontier probability mass: different parameter values can fit the estimated central moments similarly while implying different mean deviations. Motivated by the assignment at the frontier assumption, we impose a near-frontier mass constraint requiring the fitted distribution to place at least a minimum probability mass in a neighborhood of zero, with the threshold decreasing with the local effective sample size. This regularization may introduce bias when the true near-frontier probability mass falls short of the imposed lower bound, but it reduces sensitivity of the implied mean deviation to estimation noise in the central moments. In the simulations, the unconstrained estimator can have explosive mean squared error (MSE). In the empirical application, it produces implausibly large mean inefficiency estimates.

Often, research interest lies not in the frontier itself, but in deviations from it: thus, firm inefficiency in the regulated utility sector \citep{knittel2002alternative}; regulatory tax measurement in the housing market \citep{BenMosheGenesoveRegulation,glaeser2005manhattan}; cross-country differences in factor use as indicators of misallocation \citep{hsieh2009misallocation}; and firm markups \citep{loecker2012markups,de2020rise,hall1988relation}. In DEA, empirical work usually reports efficiency scores and peer benchmarks rather than characterizing the frontier \citep{banker1984some,fare1994productivity}; in SFA, the primary objective is usually to measure firm-level inefficiency, with the frontier estimated to separate it from random error \citep{jondrow1982estimation,battese1995model}. When the assignment at the frontier assumption holds and the frontier is estimable, deviations can be recovered from differences between the estimated frontier and observed outcomes.

However, assignment at the frontier need not hold. For example, if firms that choose inputs suboptimally are also always inefficient in applying those inputs, then the upper envelope of observed outputs will lie strictly below the frontier production function, touching it only at cost-minimizing input vectors chosen by efficient firms.\footnote{For strictly quasi-concave and monotone production functions, cost minimization under unconstrained input choice and given prices implies a unique input vector. Thus, identification of the frontier production function over some region of inputs requires: (a) variation in input prices, (b) adjustment costs for inputs, or (c) firms that are efficient in applying inputs making suboptimal input choices. These are the same sources for input variation as in mean regression.} More generally, firms that choose inputs suboptimally but apply those inputs efficiently may exist, but may be rare at any given input vector. Near-frontier observations may therefore be scarce, especially at input vectors far from those of cost-minimizing efficient firms.

Scarcity of near-frontier data motivates our focus on bounding the mean deviation. We derive a lower bound on the mean deviation using only variance and skewness. Because it depends only on the second and third central moments, it can remain informative even when there are no observations near the frontier. This contrasts with frontier-based methods, which must rely on extrapolation based on distributional assumptions. Regularization is attractive when a parametric specification is a reasonable approximation and the near-frontier mass constraint is plausible. The lower bound is a robust complement: it imposes no parametric distributional assumptions and remains valid even when scarcity causes the near-frontier mass constraint to fail.

We apply our results to estimate a production function, which is usually estimated using conditional mean regression. To address the main challenge of endogeneity, the now standard approach uses a control function with a proxy variable to capture unobserved productivity \citep{olleypakes,levinsohn2003estimating}. In contrast, our frontier approach identifies the frontier production function by exploiting two assumptions: inefficiency is nonnegative and zero lies in its support given inputs. This identification strategy requires neither proxy variables nor instrumental variables. 
Our approach allows the distribution of inefficiency to be statistically dependent on inputs, so inputs need not be exogenous. Allowing this dependence is central to our approach: if inefficiency were mean-independent of inputs, the frontier would differ from the conditional mean only by a constant shift.

Our analysis builds on the SFA literature \citep[see, e.g., the review by][]{kumbhakar2022stochastic}, which originally assumed a linear frontier production function and mutually independent inefficiency, random error, and inputs \citep{aigner1977formulation,meeusen1977efficiency}. Since then, SFA models have incorporated more flexible frontier functions \citep[e.g.,][]{fan1996semiparametric} and distributions \citep[e.g.,][]{reifschneider1991systematic,parmeter2019combining}.

Given the centrality of endogeneity in econometrics, growing attention has focused on it in frontier estimation \citep[e.g.,][]{amsler2016endogeneity,karakaplan2017endogeneity,prokhorov2021estimation,centorrino2023maximum}. A central contribution of our paper is to show that common solutions, such as instrumental variables, are often superfluous. In fact, the frontier is identified by nonnegative inefficiency and the assignment at the frontier assumption, provided that the distribution of inefficiency is allowed to vary with inputs. Some SFA models permit dependence between inefficiency and inputs, often to model the determinants of inefficiency. However, to our knowledge, the literature has not recognized that allowing for this dependence is what accounts for endogeneity in frontier estimation.

A potential outcomes framework can elucidate why endogeneity is not a concern under assignment at the frontier and how this contrasts with random assignment. Let $\tau \in [0,1]$ index unobserved firm types, with potential outcomes $Y_0(\tau)$ under control and $Y_1(\tau)$ under treatment, each strictly decreasing in $\tau$ and with finite suprema at the frontier type $\tau=0$. The observed outcome is $Y = D \cdot Y_1(\tau) + (1-D) \cdot Y_0(\tau)$, where $D$ indicates treatment. Under random assignment, $\tau$ and $D$ are independent, and the average treatment effect is identified by the difference in conditional means, $E[Y \mid D = 1] - E[Y \mid D = 0]$.  Alternatively, under assignment at the frontier, $\tau=0$  is in the support under both treatment and control, and the frontier treatment effect is identified by the difference in conditional suprema, $\sup(Y \mid D = 1) - \sup(Y \mid D = 0)$.  Random assignment requires the distribution of types to be the same across treatment and control. Assignment at the frontier requires that the potential outcomes have finite suprema and that every neighborhood of the frontier type has strictly positive probability under both groups. The average treatment effect and the frontier treatment effect may differ substantially, and the relevant question determines which is of primary interest.\footnote{For policies aimed at the general population, the average effect is often the relevant object. For policies concerned with best practice or maximum potential, where a natural frontier exists, the frontier effect may be more relevant.}


Our approach can be compared to the identification at infinity literature, which exploits tail behavior for identification---using large values of covariates or instruments \citep{lewbel2007endogenous,chamberlain1986asymptotic} or large values of the outcome itself \citep{dhaultfoeuille2013another}. That work assumes the identified structural relationship in these extremes extends to average observations through a stable functional form and is usually interested in average effects. By contrast, our interest is in the frontier (e.g., the frontier production function) and the economically meaningful deviations from it (e.g., firm inefficiency). Our identification does not rely on extreme inputs or extreme unconditional outcomes; instead, it relies on extreme outcomes conditional on inputs.

Monte Carlo simulations assess finite-sample performance of our moment-based estimator of mean deviation and of the skewness-based lower bound. Two issues arise in small samples: the estimated lower bound can exceed the estimated mean, and the estimated mean can be unstable even when the lower bound remains accurate. The near-frontier mass constraint substantially reduces MSE. It rarely binds when the underlying distribution has high near-frontier mass, indicating it does not mechanically override moment information.

For illustrative purposes, we apply the methods to plant-level data from Colombian manufacturing in the food products industry.\footnote{These data have been widely used in production applications \citep[e.g.,][]{eslava2004effects,fernandes2007tradepolicy}. We use the version provided in the \texttt{gnrprod} \texttt{R} package for \citet{gandhi2020identification}.} Estimated mean inefficiency is correlated with inputs,  whereas standard OLS and SFA rule out such dependence by assumption. In the absence of regularization, unconstrained central moment matching can produce fitted inefficiency distributions with little mass near the frontier and implausibly large mean inefficiency estimates. The regularized estimator yields results that capture the overall shape of the empirical distribution while satisfying the near-frontier mass constraint.

The remainder of the paper is organized as follows. Section~\ref{se:model} presents the structural model and defines the FSF and mean deviation.
Section~\ref{se:id} introduces the assignment at the frontier assumption and identifies the FSF.
Section~\ref{se:ASF} provides conditions under which the FSF recovers the ASF and structural function.
Section~\ref{se:modelerror} introduces random errors.
Section~\ref{se:bound} derives lower bounds for mean deviation using variance, skewness, and higher central moments via the Stieltjes moment problem.
Section~\ref{se:est} presents semiparametric method of moments estimation with regularization using the near-frontier mass constraint in a panel data setting with time-invariant deviation.
Section~\ref{se:sim} presents simulations.
Section~\ref{se:ap} presents the empirical application to production.
Section~\ref{se:con} concludes.

\section{The Model} \label{se:model}

Consider the structural model
\begin{align}
	y &= \widetilde g(x,\omega),  \label{eq:modelnonsep1}  
\end{align}
where $y$ is an observed scalar outcome, $x$ is a vector of observed inputs that may be arbitrarily dependent on the vector of unobservables $\omega$, and $\widetilde g(\cdot)$ is an unknown structural function.
 
Define the FSF as the supremum of the structural function given $x$, and assume this supremum is finite:
\begin{align}
	g(x) := \sup_{\omega'} \widetilde g(x,\omega') < \infty .\label{eq:FSF}
\end{align}
The supremum need not be attained. When it is, it may be achieved by multiple $\omega$-types for a given $x$, and this set of optimal types can vary across $x$. For example, in a production setting with $\omega$ representing managerial efficiency, different managers could be equally optimal for a given set of inputs, and the most effective manager type could change with the firm's input mix. The substantive assumption, however, is that the supremum itself is finite. For example, a firm’s maximum possible output must be constrained by some technological barrier, given current know-how.

Rewrite the structural model \eqref{eq:modelnonsep1} as the frontier model
\begin{align}
	y &= g(x) - u, \label{eq:model1} \\
	u &\geq 0, \label{eq:model2}
\end{align}
where $u := g(x)  - \widetilde g(x,\omega) \geq 0$ is the unobserved deviation from the FSF ($u=0$ if and only if $\omega \in \arg\max_{\omega'} \widetilde g(x,\omega') $). This framework allows an unrestricted joint distribution of $(u,x)$, and hence endogenous $x$, requiring only $u \geq 0$. 

In addition to the FSF in \eqref{eq:FSF}, we consider the mean deviation:
\begin{align*}
	 E[u \mid x] = g(x) - E[y \mid x].
\end{align*}
For example, in a production setting, for given inputs the FSF equals the frontier production function and the mean deviation equals mean inefficiency.


\subsection{Assignment at the Frontier and Identification} \label{se:id}

We assume that zero is in the support of the deviation, given inputs.

\vspace{0.2cm}
\begin{asn}[Assignment at the Frontier]\label{asn:AAB}
	\begin{align*}
		0 \in \Supp(u \mid x).
	\end{align*}
\end{asn}
\vspace{0.2cm}

In the context of the structural model \eqref{eq:modelnonsep1}, Assumption \ref{asn:AAB} requires that at least one type $\omega \in \arg\max_{\omega'} \widetilde g(x,\omega')$ lies in the support. This is the formal counterpart to the condition in the potential outcomes framework that the frontier type $\tau=0$ lies in the support. Beyond requiring support at zero (equivalently, $\inf(u\mid x)=0$ since $u\ge 0$), the assumption imposes no restriction on how the distribution of $u$ varies with $x$.\footnote{If the infimum deviation is a constant $u_{\min}>0$ common to all inputs, the FSF is identified only up to a parallel shift. This could arise, for example, if deviations reflect regulatory taxes with a uniform minimum level. If instead the infimum varies with inputs, not even the shape of the FSF is identified.}  For example, if $u$ measures inefficiency, the assumption states that given inputs $x$, there exist firms arbitrarily close to full efficiency ($u=0$), though the density near zero may vary with $x$. 



Violations of the assignment at the frontier assumption may be empirically observable. Economic theory can impose shape constraints on the frontier (e.g., monotonicity, concavity, or homogeneity), so an estimated frontier that violates any of these (e.g., a downward sloping supply curve or increasing returns where only constant or decreasing returns are admissible) suggests the frontier is not being observed in that region. 

 The FSF is identified by the supremum of the outcome given inputs.
\vspace{0.2cm}
\begin{prop}
	Assume that \eqref{eq:model1}--\eqref{eq:model2} and Assumption \ref{asn:AAB} hold.
	Then $g(x)$ is identified.
\end{prop}

\begin{proof}
	By Assumption \ref{asn:AAB} ($0\in\Supp(u\mid x)$) and \eqref{eq:model2} ($u\ge 0$), it follows that $\inf(u\mid x)=0$. Using \eqref{eq:model1} ($y=g(x)-u$), we have
	\begin{align*}
		g(x) &= g(x) - \inf(u \mid x) = \sup\big(g(x)-u \mid x\big) = \sup(y \mid x). \qedhere
	\end{align*}
\end{proof}

Thus the assignment at the frontier assumption obviates the need for instrumental variables or other corrective measures for endogeneity.

The above identification argument follows a structure similar to that of mean regression under conditional mean independence:	If $y=g(x)-\eps$ and $E[\eps \mid x]=0$, then
	\begin{align*}
		g(x) &= g(x) - E[\eps \mid x] = E[g(x) - \eps \mid x] = E[y \mid x]. 
	\end{align*}
This same logic also applies when identification relies on instrumental variables.\footnote{Suppose $y=g(x)-\eps$ and there exists exogenous variable $z$ such that $E[\eps \mid z] = 0$. Then $g(x)$ is identified from $E[g(x) \mid z] = E[y \mid z]$, under a completeness condition (i.e., the mapping $g \mapsto E[g(x) \mid z]$ is injective) \citep{newey2003instrumental}.} 

However, the assignment at the frontier and mean independence assumptions are conceptually different. Assignment at the frontier assumes that the support of $u \mid x$ includes zero, so identification does not require exogenous inputs, instruments, or any orthogonality conditions. Moreover, it is a pointwise support condition: at each input value $x$ where identification is claimed, $0 \in \Supp(u \mid x)$. This is an economically substantive restriction. It requires that $g(x)$ be attainable arbitrarily closely at that specific input vector, rather than being ruled out there by inefficiency, distortions, or other frictions. Mean independence, by contrast, assumes an orthogonality condition, here that $E[\eps \mid x]=0$, and places no restrictions on the support of $\eps \mid x$, including its behavior in neighborhoods of zero. Its identifying content comes from cross-$x$ restrictions on conditional means, rather than from a pointwise support condition.

\subsection{Relating the FSF to the ASF and Structural Function} \label{se:ASF}
 
In general, unless strong restrictions are imposed on the structural function $\widetilde g(\cdot)$ in \eqref{eq:modelnonsep1}, the FSF does not recover the ASF or the structural function.\footnote{The ASF is a central object of interest in the control function literature, defined as $ASF(x) := E_{\omega}[\widetilde g(x,\omega)]$. It represents the average outcome for a given $x$ if the endogeneity problem were absent, as it integrates over the marginal distribution of the unobservable $\omega$ rather than the conditional distribution used to compute $E[y \mid x]$.} Nevertheless, there are cases where it does. Specifically, the FSF recovers the ASF under additive separability, and it recovers the full structural function when inefficiency acts solely by rescaling inputs. However, in other specifications, such as models with interactions between unobservables and inputs, the FSF generally fails to identify either the ASF or the structural function. The following three examples in a production-function setting illustrate these different identification cases.

First, consider the model in which inefficiency is additively separable: a type-$\omega$ firm has production function  $g(x) - \omega$, where $\omega \geq 0$. In this case, the FSF identifies the function $g(\cdot)$ by observing the most efficient firms on the frontier (where $\omega=0$). However, the mean regression function, $E[y \mid x] =g(x) - E[\omega \mid x]$, fails to identify $g(x)$ if input choices are endogenous (that is $E[\omega\mid x]$ depends on $x$). Because the ASF is $g(x) - E[\omega]$, the FSF therefore identifies the ASF as well.\footnote{Formally, for $\widetilde g(x,\omega) = g_1(x) - g_2(\omega)$ with $g_2(\omega)\geq g_2(0)$, the FSF is $g_1(x) - g_2(0)$ and the ASF is $g_1(x) - E[g_2(\omega)]$. Their difference, $c = E[g_2(\omega)] - g_2(0)$ is identified by $E[FSF(x)- y]$, provided $FSF(x)$ is identified over the support of $x$.}  
		
Second, consider a model in which inefficiency rescales inputs: a type-$\omega$ firm has production function $g(x e^{-\omega})$, where $x e^{-\omega}$ are inefficiency-adjusted inputs and $\omega$ is a vector with nonnegative components. The structural function is thus $\widetilde g(x, \omega) = g(x e^{-\omega})$. This means the output of a firm with inputs $x$ and inefficiency $\omega$ equals that of an efficient ($\omega=0$) firm with inputs rescaled to $x e^{-\omega}$. In this case, the FSF identifies the entire structural function. Once the function $g(\cdot)$ is obtained from the frontier, the structural outcome for any $(x, \omega)$ pair is known simply by evaluating the frontier function at the rescaled inputs.\footnote{Formally, the FSF identifies the function $g(\cdot)$ because $FSF(x) = \widetilde g(x,0) = g(xe^{-0}) = g(x)$. The structural function for any $(\omega,x)$ is $\widetilde g(x, \omega) = g(xe^{-\omega})$. Therefore, the structural function is simply the identified frontier function evaluated at the rescaled input vector, i.e., $\widetilde g(x, \omega) = FSF(xe^{-\omega})$.}  However, the ASF is not identified in general: a scalar outcome $y$ is not sufficient to uniquely determine the multidimensional inefficiency vector $\omega$ (and thus not its distribution). An exception is the single-input case with a strictly monotonic frontier $g$, which allows $\omega$ to be recovered from the observed data as $\omega = \ln(x) - \ln(g^{-1}(y))$.

Finally, consider a model in which inefficiency interacts with inputs: a type-$\omega$ firm has production function $y=\beta x-\alpha x\omega-\omega$, where $x,\beta,\alpha,\omega\ge 0$. The marginal productivity of such a firm is $\beta-\alpha\omega$. However, consider the alternative model in which a type-$\omega^*$ firm has production function $y=\beta x-\omega^*$, where $\omega^*:=(\alpha x+1)\omega$. The marginal productivity of such a firm is $\beta$.  The two models are observationally equivalent in $(x,y)$, but differ in that the first assumes that $\omega$ is invariant to a change in $x$, whereas the second assumes that it is $\omega^*$ that is invariant.  The models share the same $FSF(x)=\beta x$ but differ in the ASF. While the frontier is identified by observing efficient firms with $\omega=0$, the interaction parameter $\alpha$ is not.\footnote{More generally, for any alternative parameter $\alpha^* \ge 0$, the model $y=\beta x-\alpha^*x\omega^*-\omega^*$ with $\omega^*:=\frac{\alpha x+1}{\alpha^*x+1}\omega$ has interaction $\alpha^*$ yet generates the same $(x,y)$ distribution, with mean marginal productivity $\beta-\alpha^*E[\omega^*\mid x]$ rather than $\beta-\alpha E[\omega\mid x]$. Consequently, $\alpha$ is not identified.}

\subsection{Random Errors} \label{se:modelerror}


Consider the frontier model \eqref{eq:model1}--\eqref{eq:model2} with random errors in a panel setting with time-invariant deviations, so that for any given firm
\begin{align}
	y_{t} &= g(x_{t}) - u + v_{t}, \quad t=1,\ldots,T, \label{eq:modelerror1}\\
	u &\geq 0, \label{eq:modelerror2}
\end{align}
where $u$ is a time-invariant deviation and $v_{t}$ is a random error.

Assume that the errors are mean independent of the inputs and that errors and deviation are mutually independent conditional on inputs.
\vspace{0.2cm}
\begin{asn}\label{asn:v}  \
	\begin{enumerate}[(i)]
		\item $E[v_t \mid x_1,\ldots,x_T] = 0$ for all $t$;
		\item $u \indep (v_1,\ldots,v_T) \mid (x_1,\ldots,x_T)$;
		\item $(v_1,\ldots,v_T)\mid(x_1,\ldots,x_T)$ are mutually independent and identically distributed.
	\end{enumerate}
\end{asn}
\vspace{0.2cm}

The independence in Assumption~\ref{asn:v}(ii) holds only conditional on $(x_{1},\ldots,x_{T})$, allowing dependence between $u$ and $\{v_t\}_{t=1}^T$ through the (potentially endogenous) inputs. Conditional on $(x_1,\ldots,x_T)$, the distribution of $y_t$ is the convolution of the distribution of $g(x_t)-u$ with that of $v_t$.

In a cross-sectional setting (i.e., $T=1$), identification of $g$ relies on differences in the supports of $u$ on $[0,\infty)$ and $v$ on $(-\infty,\infty)$, often exploiting asymmetry of $u$ and symmetry of $v$.\footnote{\cite{schwarz2010consistent} show identification when $u$ has gaps in its support and $v \sim N(0,\sigma^2)$ with unknown $\sigma^2>0$. \cite{delaigle2016methodology} show identification when $u$ is symmetric and indecomposable and $v$ is symmetric. \cite{bertrand2019flexible} show identification when $u$ is bounded and $v \sim N(0,\sigma)$. \cite{florens2020estimation} show identification when $\Supp(u) = [0,\infty)$ and $\Supp(v) = (-\infty,\infty)$ with $v$ symmetric.} In contrast, panel data (i.e., $T \geq 2$) lets us separate the central moments of $u$ and $v_t$ using within- and between-firm variation, through within-firm differencing and between-firm averaging. For the procedures below, we therefore only require identification of the second through fourth conditional central moments (Appendix~\ref{ap:estdetails}). 

\subsection{Bounding Mean Deviations}  \label{se:bound}

When data near the frontier are scarce, direct estimation of the frontier is infeasible from the data alone; instead, any estimate must rely on extrapolation based on distributional assumptions. This undermines precise estimation of the mean deviation based on the difference between an estimated frontier and the observed outcome. However, the mean deviation can still be bounded from below by a method that is robust to scarcity of data near the frontier. In this section, we use the Stieltjes moment problem to derive a family of lower bounds on the mean of a nonnegative random variable using only central moments. One bound in this family that uses only variance and skewness is especially useful in practice, as these moments are often reliably estimable even when data are scarce near the frontier, thereby avoiding strong assumptions about the distribution of the deviation near zero. These bounds provide no meaningful information about $g(x)$, as they impose no constraints on its derivatives.

Moment problems seek necessary and sufficient conditions for a sequence of raw moments $(m_j)_{j=1}^\infty$, where $m_j = E[u^j]$, to correspond to the distribution of a random variable~$u$ \citep{schmudgen2017moment}.\footnote{Recent work uses moment problems in fixed-effects logit models \citep{davezies2025identification,dobronyi2021identification}.} The Stieltjes moment problem adds the restriction that the support of $u$ is contained in the nonnegative real line.  Let $\Delta_n^{(0)}$ and $\Delta_n^{(1)}$ denote the $n\times n$ primary and shifted Hankel matrices:
{\begin{align*}
\Delta_n^{(0)} = 
\begin{pmatrix}
	1      & m_1       & \cdots & m_{n-1} \\
	m_1    & m_2       & \cdots & m_{n}   \\
	\vdots & \vdots & \ddots & \vdots  \\
	m_{n-1}& m_{n}  & \cdots & m_{2n-2}
\end{pmatrix},
\quad
\Delta_n^{(1)} = 
\begin{pmatrix}
	m_1    & m_2      & \cdots & m_{n}   \\
	m_2    & m_3    & \cdots & m_{n+1} \\
	\vdots & \vdots  & \ddots & \vdots  \\
	m_{n}  & m_{n+1}& \cdots & m_{2n-1}
\end{pmatrix}.
\end{align*}
}A well-known result states that if $\Supp(u)\subseteq[0,\infty)$ and 
\[
 \det\bigl(\Delta_n^{(0)}\bigr) \geq 0,\quad \det\bigl(\Delta_n^{(1)}\bigr) \geq 0
\quad\text{for all } n \in \mathbb{N},
\]
then there exists at least one distribution on $[0,\infty)$ with those moments. If any determinant is negative, no such distribution exists.  

In Appendix \ref{ap:primaryhankel}, we show that nonnegativity of the primary Hankel determinants implies well-known inequalities, such as nonnegative variance and that kurtosis exceeds the squared skewness by at least one. These restrictions depend only on central moments and not $E[u]$, so they do not produce bounds on $E[u]$.

Next, we turn to the nonnegative determinants of the shifted Hankel matrices. 
These produce inequalities in the raw moments $m_j$ that, once re-expressed in terms of $E[u]$ and central moments, yield bounds on $E[u]$. The first three are:
\begin{align}
	0 &\leq m_1, \label{eq:Mean}\\
	0 &\leq m_1 m_3 - m_2^2, \label{eq:quad}\\
	0 &\leq m_1(m_3m_5 - m_4^2) - m_2(m_2m_5 - m_3m_4) + m_3(m_2m_4 - m_3^2). \label{eq:cubic}
\end{align}
Expressing raw moments in terms of $E[u]$ and central moments $\mu_j=E[(u - E[u])^j]$, each condition $\det\bigl(\Delta_n^{(1)}\bigr) \geq 0$ becomes an $n$th degree polynomial inequality in $E[u]$ with coefficients that are polynomials in $\mu_2,\ldots,\mu_{2n-1}$. 
Equation \eqref{eq:Mean} states only that $E[u]\geq0$. 
We now state a skewness-based lower bound for $E[u]$, obtained by solving \eqref{eq:quad} for $E[u]$.

\begin{prop}\label{prop:skew-bound}
	\begin{align}
		E[u] \geq \frac{-\mu_3 + \sqrt{\mu_3^2 + 4\mu_2^3}}{2\mu_2}
		=\frac{\sigma}{2}\Bigl(-\gamma + \sqrt{\gamma^2 + 4}\Bigr). \label{eq:boundu}
	\end{align}
	where $\sigma = \sqrt{\mu_2}$ denotes the standard deviation and $\gamma=\mu_3/\mu_2^{3/2}$ denotes the skewness.
\end{prop}
\begin{proof}
	Expanding \eqref{eq:quad} gives a quadratic in $E[u]$:
	$0 \le \mu_2(E[u])^2+\mu_3E[u]-\mu_2^2$.
	Since $\mu_2>0$ and $E[u]\ge 0$ by \eqref{eq:Mean}, this holds if and only if $E[u]$ is at least the larger root. \qedhere
\end{proof}

The right-hand side of \eqref{eq:boundu} is nonnegative and strictly decreasing in $\gamma$, converging to $0$ as $\gamma\to\infty$.

If one is prepared to assume that $u$ is negatively skewed, then the proposition implies that the standard deviation provides a lower bound on the mean deviation. This is useful for inferring bounds on mean inefficiency from reported empirical results on heterogeneity in firm cost or productivity.\footnote{For example, \cite{syverson2004market} reports a 0.34 standard deviation in TFP of plants in high construction markets, with a pictured probability density distribution that shows no substantial right-skewness (pages 1184 and 1185).  The bound then allows us to reasonably conclude that mean inefficiency is at least 0.34. 
	In contrast,  \cite{foster2008reallocation}, who report a 0.22 mean standard deviation of TFP across a large set of industries, provide no information on higher moments that might allow us to bound mean inefficiency.} 
Specifically, if $\gamma\leq 0$ then 
\begin{align*}
	E[u] \geq \frac{\sigma}{2}\Bigl(-\gamma + \sqrt{\gamma^2 + 4}\Bigr)
	\geq \sigma.  
\end{align*}

Expanding \eqref{eq:cubic} gives a cubic inequality in $E[u]$.\footnote{The full expansion is
	\begin{align*}
		0 \le & (\mu_2 \mu_4-\mu_2^3-\mu_3^2)(E[u])^3
		+(\mu_2 \mu_5-\mu_2^2\mu_3-\mu_3\mu_4)(E[u])^2 
		+(\mu_3\mu_5-\mu_4^2-\mu_2\mu_3^2+\mu_2^2\mu_4)E[u]
		+(2\mu_2\mu_3\mu_4-\mu_3^3-\mu_2^2\mu_5).
	\end{align*}
}
If $\mu_3\leq 0$ and $\mu_5\leq 0$ then
$$
	E[u]\geq\sqrt{\frac{\mu_4}{\mu_2}}=\sigma \times \sqrt{\kappa},
$$
where $\kappa= \mu_4/\mu_2^2$ denotes the kurtosis.

Higher-order determinants produce polynomial inequalities in $E[u]$ that may yield tighter bounds than the skewness-based bound in \eqref{eq:boundu}. While we have clear intuition for variance and skewness, higher-order moments, such as the fifth central moment, are less transparent in their interpretation. Moreover, they are challenging to estimate accurately due to their sensitivity to outliers, which is why higher-order moments are rarely employed in practice.


\section{Estimation} \label{se:est}

Estimating a frontier is challenging. Under certain tail behaviors, extreme value theory implies logarithmic convergence rates. For example, \citet{goldenshluger2004bndry} assume normally distributed errors and show that the best-possible convergence rate for frontier estimation is logarithmic. Nonparametric methods for estimating the densities of $u$ and $v$ via deconvolution of \eqref{eq:modelerror1} have not been widely adopted in practice, as they require tuning parameters, involve substantial computational complexity, and converge slowly. Consequently, researchers often impose parametric assumptions to simplify estimation and improve convergence rates.

When $x$ is continuous, estimation becomes even more demanding, as it often requires smoothing in $x$ or restricting how the distributions of $u$ or $v$ vary with $x$. These restrictions can be problematic, especially when one wishes to allow for endogeneity. These challenges motivate two broad estimation strategies in the SFA literature: maximum likelihood \citep[e.g.,][]{aigner1977formulation} and moment-based estimators \citep[e.g., OLS-based corrections such as COLS,][]{olson1980monte}. We adopt the latter in a panel-data setting, nonparametrically estimating the conditional mean outcome and second through fourth conditional central moments of deviations and errors. Appendix~\ref{ap:estdetails} collects moment identities and sample formulas used for the panel estimator, and Appendix~\ref{ap:datastruct} discusses related estimators under alternative data structures. 

In earlier versions of the paper, we also implemented a cross-sectional variant (Appendix~\ref{ap:cs}). In Monte Carlo simulations and in the empirical application, the cross-sectional moment-matching objective often exhibited multiple well-fitting local optima. With only second through fourth central moments, tail thickness can trade off between the deviation and the random error, so different decompositions can fit the data similarly while attributing heavy tails to different components. This is consistent with \citet{parmeter2024inference}, who show in cross-sectional Monte Carlo simulations that specification tests can have empirical rejection rates that deviate substantially from nominal size in finite samples, even under correct parametric specification and when optimization is initialized at the true parameter values.

Consider the model \eqref{eq:modelerror1}--\eqref{eq:modelerror2} with panel data $\{y_{it},x_{it}\}$ and time-invariant deviations:
\begin{align}
	y_{it} &= g(x_{it}) - u_i + v_{it}, \quad t=1,\ldots,T_i,\ i=1,\ldots,n,
	\label{eq:model_fe1}\\
	u_i &\ge 0.
	\label{eq:model_fe2}
\end{align}

Suppose Assumption \ref{asn:v} holds and let $x_i=(x_{i1},\ldots,x_{iT_i})'$.  Our identification strategy relies on decomposing the variation in $y_{it}$ into ``within'' and ``between'' components to separate the central moments of the random error $v_{it}$ from the time-invariant deviation $u_i$. Define the population centered residuals as $\eps_{it} = y_{it} - E[y_{it} \mid  x_i]=-(u_i-E[u_i\mid   x_i])+v_{it}$. We decompose these residuals into:
\begin{align}
	\text{Within:} \quad \eps_{it}^w &= \eps_{it} - \bar \eps_i = v_{it} - \bar v_i, \label{eq:within_id} \\
	\text{Between:} \quad \bar \eps_i &= \frac{1}{T_i}\sum_{t=1}^{T_i} \eps_{it}= -( u_i - E[u_i \mid  x_i]) + \bar v_i, \label{eq:between_id}
\end{align}
where $\bar v_i=T_i^{-1}\sum_{t=1}^{T_i} v_{it}$ is the time-averaged error. Equation \eqref{eq:within_id} shows that the within variation depends only on $v$, allowing identification of the conditional error central moments $\mu_{k,v}(x_i)$.\footnote{For a random variable $w$, let $\mu_{k,w}(z)=E[(w-E[w\mid z])^k\mid z]$ denote the $k$th central moment conditional on $z$.} Equation \eqref{eq:between_id} contains $u$ and the averaged $v$, so that once the central moments of $v$ are identified, they can be subtracted from the central moments of $\bar \eps_i$ to recover the central moments of $u$.

Estimation proceeds in three stages. First, we estimate the conditional expectation $E[y_{it} \mid  x_i]$ using a flexible nonparametric regression (where $	E[v_{it}\mid  x_i]=0$ by Assumption \ref{asn:v}(i)). This yields the sample residuals $\hat \eps_{it}$, which we decompose into sample within residuals $\hat \eps_{it}^w$ and sample between residuals $\bar{\hat \eps}_i$ (analogous to \eqref{eq:within_id} and \eqref{eq:between_id}).

Second, we estimate the conditional central moments of the unobserved components. Using Assumption~\ref{asn:v}(iii), we relate the empirical central moments of the within residuals $\hat \eps_{it}^w$ to the error central moments and solve for $\widehat\mu_{k,v}( x_i)$. Then, using Assumption \ref{asn:v}(ii), we construct bias-corrected powers of the between residuals $\widehat u_i^k$. We construct these by taking the powers $(\bar{\hat \eps}_i)^k$ and removing the contribution of the error term $\bar v_i$ using the estimates $\widehat\mu_{k,v}( x_i)$. Following the literature on variance estimation under heteroskedasticity \citep[e.g.,][]{hall1989variance,fan1998efficient}, we regress $\widehat u_i^k$ on $ x_i$ to obtain smoothed estimates of the deviation central moments $\widehat\mu_{k,u}( x_i)$.\footnote{The explicit formulas for the moment conditions and bias-corrected powers are provided in Appendix~\ref{ap:estdetails} \citep[see also][]{BenMosheGenesoveMoments}.}

In the final stage, we specify a parametric family for the deviation distribution $u\mid x$, and estimate parameters by the method of moments. We match model-implied central moments to the second-stage estimates, both unconstrained and subject to the near-frontier mass constraint described next.

Assumption \ref{asn:AAB} implies that every neighborhood $[0,\delta]$ has positive conditional probability mass:
\begin{align}
	F_{u\mid x}(\delta):=\Pr(u\le \delta\mid x)>0 \qquad \text{for all } \delta>0. \label{eq:cdf neigh zero}
\end{align}
The assumption does not require a density at zero, nor does it control how quickly $F_{u\mid x}(\delta)$ shrinks as $\delta\downarrow 0$ (see Appendix~\ref{ap:AAF}). To prevent our estimation from admitting distributions with arbitrarily thin left tails near zero, we regularize by imposing a finite-sample minimum near-frontier mass constraint. For chosen constants $c>0$ and $m_0>0$, we estimate $\widehat \theta(x)$ by
\begin{align}
	\widehat \theta(x)
	& = 
	\arg\min_{\theta}\  \sum_{k\in\{2,3,4\}}  \Big(\widehat \mu_{k,u}(x) - \mu_k(x;\theta)\Big)^2  
	\label{eq:mm_objective}\\
	\text{s.t.}& \quad
	F_{u\mid x;\theta}\big(c \cdot \widehat \sigma_u(x)\big) \ge \frac{m_0}{\widehat n_{\mathrm{eff}}(x)} ,
	\label{eq:mm_constrained}
\end{align}
where $\widehat \sigma_u(x)=(\widehat \mu_{2,u}(x))^{1/2},$ $\mu_k(x;\theta)$ denotes the $k$th conditional central moment of $u\mid x$ under $\theta$, and $F_{u\mid x;\theta}$ denotes the corresponding conditional CDF.\footnote{Equivalently, the constraint can be written as $Q_{u\mid x;\theta}(m_0/\widehat n_{\mathrm{eff}}(x))\le c\cdot \widehat\sigma_u(x)$, where $Q$ is the conditional quantile function.} We then compute $\widehat E[u \mid x]= E[u\mid x;\widehat\theta(x)]$ and estimate the FSF by $\widehat g(x_{it}) = \widehat E[y_{it} \mid x_i] + \widehat E[u \mid x_i]$.

Central moments are translation-invariant: $u$ and $u+\alpha$ have identical central moments for any constant $\alpha$. In the population, nonnegativity and assignment at the frontier ($u\ge 0$ and $0\in \Supp(u\mid x)$) distinguish $u$ and $u+\alpha$ for $\alpha>0$, because $u$ has mass arbitrarily near zero whereas $u+\alpha$ does not. In finite samples, however, when data near the frontier are sparse, this distinction can be weak: neither a sample from $u$ nor a sample from $u+\alpha$ may contain observations near zero, even though their means differ by $\alpha$.\footnote{For example, let $u \sim q\cdot\mathrm{Beta}(a,b)$, where $(a,b)$ determine the shape and $q$ scales the support. When observations near zero are sparse, different values of $q$ can produce distributions with similar shapes over the region where data are observed, while implying very different support widths and hence very different means.} As a result, the moment-matching objective \eqref{eq:mm_objective} can admit parameter values that fit the estimated central moments similarly while implying very different means. The estimator can therefore fit the observed shape well while remaining weakly informative about the location of that shape relative to zero. The constraint \eqref{eq:mm_constrained} restores this finite-sample distinction by requiring mass near zero, preventing the fitted distribution from drifting away from the frontier. Indeed, as illustrated in the empirical application (Figure~\ref{fig:density_zoom}), unconstrained and constrained fits can have similar overall shape while implying very different means, because the unconstrained fit places much less mass near the frontier.\footnote{The estimator also admits a quasi-Bayesian interpretation. Moment-based estimators can be viewed as maximizing a quasi-likelihood derived from the moment distance \citep{chernozhukov2003mcmc}. The constraint is equivalent to specifying a prior with restricted support that assigns zero mass to the set of parameter values violating the boundary restriction.} Thus, \eqref{eq:mm_constrained} should be understood as finite-sample regularization, not an exogeneity assumption; it does not assume that a conditional quantile of the deviation is constant across inputs.\footnote{Some approaches estimate the frontier as an extreme conditional quantile of the outcome given inputs \citep[see, e.g.,][]{daouia2007nonparametric,daouia2010frontier}. \cite{cazals2016nonparametric} develop a nonparametric instrumental variable version to deal with endogeneity.}

Using $c\cdot \widehat\sigma_u(x)$ defines the neighborhood in local standard-deviation units, so the restriction adapts to the local dispersion of $u\mid x$. The effective sample size at $x_i$ is $\widehat n_{\mathrm{eff}}(x_i):=1/\sum_{j=1}^n w_{ij}^2$ \citep{kish1965survey}, where $w_{ij}=K(\|x_i-x_j\|/h)\big/\sum_{r=1}^n K(\|x_i-x_r\|/h)$, $K(\cdot)$ is a nonnegative kernel function with weights normalized to sum to one, and $h$ is a bandwidth. This definition implies $\widehat n_{\mathrm{eff}}(x_i)=n$ under uniform weights ($w_{ij}=1/n$), and $\widehat n_{\mathrm{eff}}(x_i)$ decreases as weights become more concentrated on observations near $x_i$. 

The threshold $m_0/\widehat n_{\mathrm{eff}}(x)$ scales with the amount of information available for learning near-zero behavior at $x$. The constraint can be rewritten as a conditional moment inequality,
\begin{align*}
\widehat n_{\mathrm{eff}}(x) \cdot F_{u\mid x;\theta}\!\big(c \cdot \widehat \sigma_u(x)\big)\ \ge\ m_0,
\end{align*}
so that $m_0$ is the required expected number of effective observations within the neighborhood $[0,c\cdot \widehat\sigma_u(x)]$. For example, if $\widehat n_{\mathrm{eff}}=25$ and $m_0=1$ then $m_0/\widehat n_{\mathrm{eff}}=0.04$, so the constraint requires at least 4\% probability mass within $[0,c\cdot \widehat\sigma_u]$, corresponding, on average, to one effective observation in that neighborhood. If $\widehat n_{\mathrm{eff}}=100$ and $m_0=1$ then $m_0/\widehat n_{\mathrm{eff}}=0.01$, yielding a weaker restriction.

Relative to cross-sectional strategies that rely on asymmetry of $u$ and symmetry of $v$ for identification, our approach exploits the panel structure: within-variation identifies the random error, while between-variation identifies the deviation. Consequently, we do not impose any symmetry restriction on $v$; provided the relevant moments exist, its distribution is otherwise left unspecified. The distribution of $u$ may be left-skewed, right-skewed, or symmetric, unlike standard SFA specifications (e.g., half-normal, truncated normal, exponential, or gamma), which impose right skewness \citep{greene2008econometric}. 


\section{Monte Carlo Simulations}\label{se:sim}

We report Monte Carlo simulations to assess the finite-sample performance of the lower bound and point estimators of $E[u]$.  Two finite-sample issues arise. First, in small samples the estimated lower bound can exceed the estimated mean, so that $\widehat{LB}>\widehat{E}[u]$. Second, the estimator of $E[u]$ can be highly inaccurate even when $\widehat{LB}$ remains accurate. Motivated by the ill-posed mapping from central moments to mean deviation, we study whether imposing a near-frontier mass constraint can regularize this mapping and improve estimation of $E[u]$.

We simulate a panel with time-invariant deviations,
\[
y_{it} = g - u_i + v_{it}, \qquad i = 1,\ldots,n,\; t = 1,\ldots,T,
\]
where the random errors are normally distributed, $\{v_{it}\}_{i=1,t=1}^{n,T} \overset{\text{iid}}{\sim} \mathrm{N}(0, \sigma_v^2)$. The deviations are generated via a rejection sampling mechanism: we draw candidate deviations $\tilde{u} \sim q\cdot\mathrm{Beta}(a,b)$ and retain them with probability $1-p$ if they fall below the $f$-th quantile and with probability 1 otherwise, repeating until $n$ deviations $\{u_i\}_{i=1}^n$ are obtained. We set $g=5$, $q=4$, and $T=8$, and the scarcity parameters to $f=0.05$ and $p=0.95$. The shape parameters $(a,b)$ vary over an equally spaced $80\times 80$ grid in $(\log a,\log b)\in[-2,2]^2$. For each configuration $(a,b)$ on the grid and each $n\in\{25,250,2\,500\}$, we generate 150 independent Monte Carlo replications of panel datasets, each of size $nT$. For each design, we set $\sigma_v$ equal to the standard deviation of $u$ from the rejection sampling mechanism.

We treat the outcomes $\{y_{it}\}_{i=1,t=1}^{n,T}$ as the only observed variables; the deviations $\{u_i\}_{i=1}^n$ and errors $\{v_{it}\}_{i=1,t=1}^{n,T}$ are treated as unobserved. For each simulated dataset, we estimate the deviation central moments $\mu_{2,u}$, $\mu_{3,u}$, $\mu_{4,u}$, the lower bound $LB$, and the mean deviation $E[u]$ using the three-step procedure in Section~\ref{se:est}. First we compute residuals and decompose them into within- and between-firm components:
\begin{align*}
	\hat \eps_{it}^w = \hat \eps_{it} - \bar{\hat \eps}_i,
	\qquad
	\bar{\hat \eps}_i = \frac{1}{T}\sum_{t=1}^T \hat \eps_{it},
	\qquad
	\hat \eps_{it} = y_{it} - \bar y, 
	\qquad
	\bar y = \frac{1}{nT}\sum_{i=1}^n\sum_{t=1}^T  y_{it}.
\end{align*}

Second, we estimate the central moments of $u$ using the closed-form estimators in \eqref{eq:mu2uhat_pool}--\eqref{eq:mu4uhat_pool} in Appendix~\ref{ap:homo}. We then estimate the lower bound by the sample analog of~\eqref{eq:boundu},
\begin{align*}
	\widehat{LB} = \frac{\widehat{\sigma}_u}{2} \Big(-\widehat{\gamma}_u + \sqrt{\widehat{\gamma}_u^2 + 4}\Big),
\end{align*}
where $\widehat{\sigma}_u = (\widehat{\mu}_{2,u})^{1/2}$ and $\widehat{\gamma}_u = {\widehat{\mu}_{3,u}}/{\widehat{\sigma}_u^3}$.

Third, to estimate the mean deviation $\widehat{E}[u]$, we fit the parameters of a $q \cdot \mathrm{Beta}(a,b)$ distribution. We estimate $(\widehat a,\widehat b,\widehat q)$ by matching the estimated central moments to the model-implied central moments while imposing the near-frontier mass constraint. With $m_0=1$ and $c\in \{0.5,1.00, \infty\}$, we solve
\begin{align*}
	(\widehat a, \widehat b, \widehat q)
	& \in\ 
	\arg\min_{a, b, q > 0}\  \sum_{k \in \{2,3,4\}} \Big(\widehat \mu_{k,u} - \mu_k(a,b,q)\Big)^2 \\
	\text{s.t.}& \quad
	F_{\mathrm{Beta}(a,b)}\!\big({c \cdot \widehat \sigma_u}/{q}\big) \ge \frac{m_0}{n},
\end{align*}
where $F_{\mathrm{Beta}(a,b)}$ is the CDF of a $\mathrm{Beta}(a,b)$ random variable on $[0,1]$. For a fixed neighborhood width $c\cdot \widehat\sigma_u$, the constraint requires at least $m_0$ expected observations in this neighborhood. The unconstrained estimator corresponds to $c=\infty$. We then compute the implied mean, $\widehat{E}[u] = \widehat q \cdot \widehat a / (\widehat a + \widehat b)$.

We partition the parameter grid into four regions defined by the qualitative shape of the distribution of $u$. The high near-frontier mass region ($a<1$) corresponds to densities that are high near zero. The unimodal right-skewed ($a\geq 1, b\geq 1, b \geq a$) and unimodal left-skewed ($a\geq 1, b\geq 1, b < a$) regions correspond to bell-shaped densities that are right-skewed (or symmetric) and left-skewed, respectively. The low near-frontier mass region ($a\ge 1,b<1$) corresponds to densities with little mass near zero, rising toward the upper end of the support.

Figure~\ref{fig:EuLB} and Table~\ref{table:EuLB} compare the unconstrained estimator $\widehat{E}[u]$ with its lower bound estimator $\widehat{LB}$. The top row of Figure~\ref{fig:EuLB} shows, for each $(a,b)$, the fraction of the 150 Monte Carlo replications in which $\widehat{LB}>\widehat{E}[u]$, a violation of the population inequality~\eqref{eq:boundu}. Table~\ref{table:EuLB} shows the fractions averaged across grid points within each region: for $n=25$, 5.70\% in high near-frontier mass, 3.49\% in low near-frontier mass, 0.74\% in unimodal right-skewed, and 1.19\% in unimodal left-skewed. As $n$ increases to $250$, these violations virtually disappear.

\begin{figure}[!b]
	\centering
	
	\begin{subfigure}{0.32\textwidth}
		\centering
		\includegraphics[width=\linewidth,trim=6 3 6 6,clip]{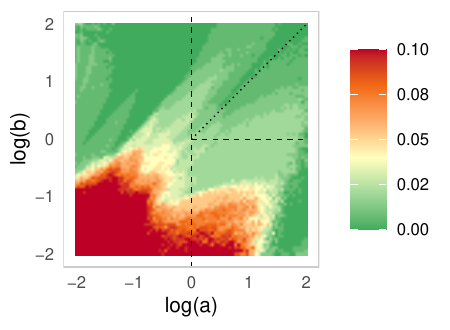}
		\caption{$n=25$, MC share $\widehat{LB} > \widehat{E}[u]$}
	\end{subfigure}\hfill
	\begin{subfigure}{0.32\textwidth}
		\centering
		\includegraphics[width=\linewidth,trim=6 3 6 6,clip]{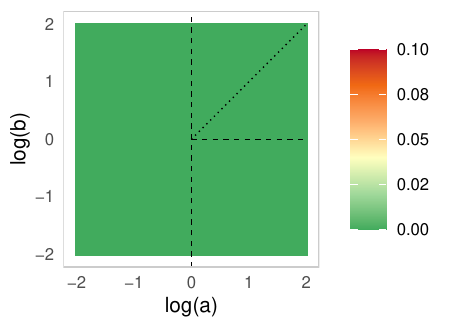}
		\caption{$n=250$, MC share $\widehat{LB} > \widehat{E}[u]$}
	\end{subfigure}\hfill
	\begin{subfigure}{0.32\textwidth}
		\centering
		\includegraphics[width=\linewidth,trim=6 3 6 6,clip]{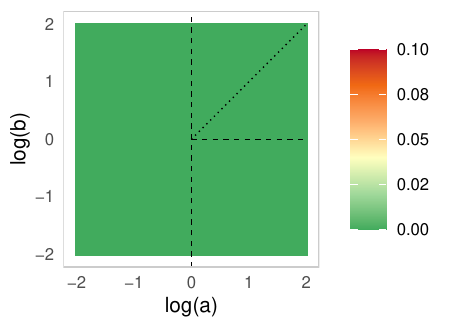}
		\caption{$n=2,500$, MC share $\widehat{LB} > \widehat{E}[u]$}
	\end{subfigure}
	
	\vspace{0.3cm}
	
	\begin{subfigure}{0.32\textwidth}
		\centering
		\includegraphics[width=\linewidth,trim=6 3 6 6,clip]{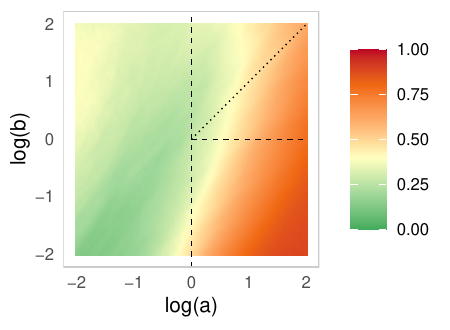}
		\caption{$n=25$, median $\frac{|\widehat{E}[u]-E[u]|}{E[u]}$}
	\end{subfigure}\hfill
	\begin{subfigure}{0.32\textwidth}
		\centering
		\includegraphics[width=\linewidth,trim=6 3 6 6,clip]{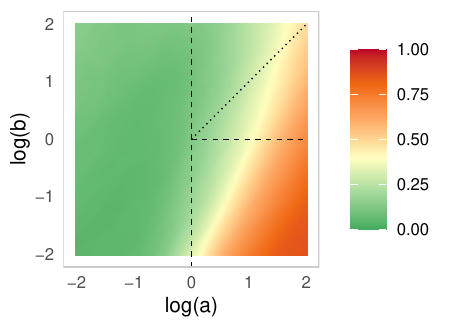}
		\caption{$n=250$, median $\frac{|\widehat{E}[u]-E[u]|}{E[u]}$}
	\end{subfigure}\hfill
	\begin{subfigure}{0.32\textwidth}
		\centering
		\includegraphics[width=\linewidth,trim=6 3 6 6,clip]{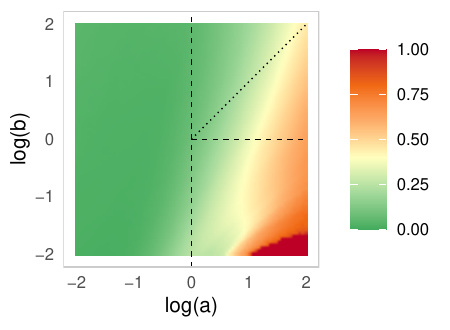}
		\caption{$n=2,500$, median $\frac{|\widehat{E}[u]-E[u]|}{E[u]}$}
	\end{subfigure}
	
	\vspace{0.3cm}
	
	\begin{subfigure}{0.32\textwidth}
		\centering
		\includegraphics[width=\linewidth,trim=6 3 6 6,clip]{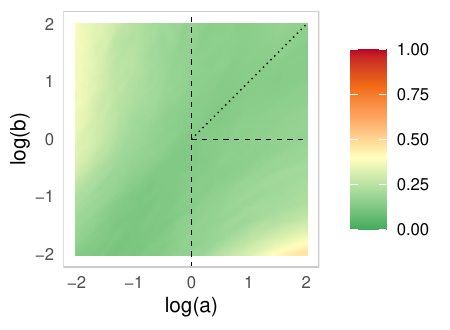}
		\caption{$n=25$, median $\frac{|\widehat{LB}-LB|}{LB}$}
	\end{subfigure}\hfill
	\begin{subfigure}{0.32\textwidth}
		\centering
		\includegraphics[width=\linewidth,trim=6 3 6 6,clip]{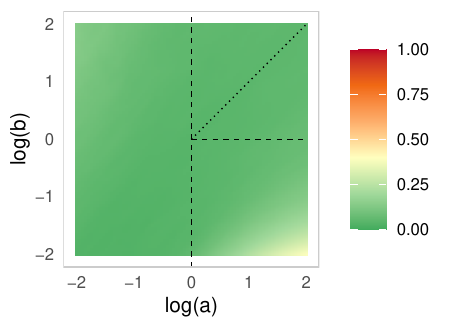}
		\caption{$n=250$, median $\frac{|\widehat{LB}-LB|}{LB}$}
	\end{subfigure}\hfill
	\begin{subfigure}{0.32\textwidth}
		\centering
		\includegraphics[width=\linewidth,trim=6 3 6 6,clip]{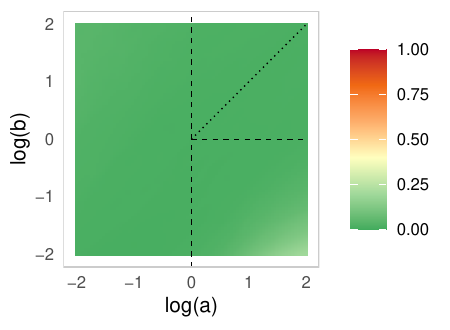}
		\caption{$n=2,500$, median $\frac{|\widehat{LB}-LB|}{LB}$}
	\end{subfigure}
	
	\caption{
		Heatmaps over the parameter grid $(\log a,\log b)$. Columns correspond to sample sizes $n\in\{25,250,2\,500\}$. The top row shows the fraction of Monte Carlo replications in which $\widehat{LB}>\widehat{E}[u]$. The middle row shows the median relative absolute error of $\widehat{E}[u]$, and the bottom row shows the median relative absolute error of $\widehat{LB}$, each computed across the 150 Monte Carlo replications at a given grid point.
	}
	\label{fig:EuLB}
\end{figure}

\begin{table}[!t]
	\centering
	\footnotesize
	\vspace{0.1cm}
	\setlength{\tabcolsep}{2.5pt}
	\renewcommand{\arraystretch}{0.9}
	\begin{threeparttable}
		\caption{Comparison of Mean and Lower-Bound Estimates by Region}
		\label{table:EuLB}
		\begin{tabular}{@{}l cccc cccc cccc@{}}
			\toprule
			& \multicolumn{4}{c}{MC share $\widehat{LB} > \widehat{E}[u]$}
			& \multicolumn{4}{c}{Median $\frac{|\widehat{E}[u] - E[u]|}{E[u]}$}
			& \multicolumn{4}{c}{Median $\frac{|\widehat{LB} - LB|}{LB}$} \\
			\cmidrule(lr){2-5} \cmidrule(lr){6-9} \cmidrule(lr){10-13}
			$n$ & High-NFM & Uni-R & Uni-L & Low-NFM & High-NFM & Uni-R & Uni-L & Low-NFM & High-NFM & Uni-R & Uni-L & Low-NFM \\
			\midrule
			25    & 0.057 & 0.007 & 0.012 & 0.035 & 0.275 & 0.380 & 0.566 & 0.644 & 0.205 & 0.152 & 0.153 & 0.203 \\
			250   & 0.000 & 0.000 & 0.000 & 0.000 & 0.090 & 0.212 & 0.422 & 0.532 & 0.061 & 0.052 & 0.055 & 0.101 \\
			2,500 & 0.000 & 0.000 & 0.000 & 0.000 & 0.040 & 0.170 & 0.391 & 0.603 & 0.023 & 0.019 & 0.017 & 0.037 \\
			\bottomrule
		\end{tabular}
	\begin{tablenotes}[para,flushleft]
		\hspace*{-0.4em} Notes:
		Regions are defined by shape parameters $(a,b)$:
		High-NFM, high near-frontier mass ($a<1$);
		Uni-R, unimodal right-skewed ($1 \le a \le b$);
		Uni-L, unimodal left-skewed ($1 \le b < a$);
		Low-NFM, low near-frontier mass ($a \ge 1,\, b<1$).
		For each grid point, statistics are first computed across the 150 Monte Carlo replications.
		The reported entries then average these grid-point statistics across all grid points in the region.
		The first column group reports the mean across grid points of the fraction of replications in which the lower-bound estimate exceeds the mean estimate.
		The second and third column groups report the mean across grid points of the median relative absolute error of $\widehat{E}[u]$ and $\widehat{LB}$, respectively.
	\end{tablenotes}
	\end{threeparttable}
\end{table}

The middle and bottom rows of Figure~\ref{fig:EuLB} show, for each $(a,b)$, the median relative absolute errors ${|\widehat{E}[u]-E[u]|}/{E[u]}$ and ${|\widehat{LB}-LB|}/{LB}$. Table~\ref{table:EuLB}, averaging across grid points within each region, shows that the median relative absolute error of $\widehat{E}[u]$ is large at $n=25$ and varies substantially by region: 27.5\% in high near-frontier mass, 38.0\% in unimodal right-skewed, 56.6\% in unimodal left-skewed, and 64.4\% in low near-frontier mass. In contrast, the lower-bound estimator is more accurate and stable across regions (20.5\%, 15.2\%, 15.3\%, and 20.3\%). As $n$ increases, the median error of $\widehat{LB}$ declines in all regions (to between 1.7\% and 3.7\% at $n=2{,}500$), while the median error of $\widehat{E}[u]$ remains large in parts of the grid: at $n=2{,}500$, it is 4.0\% in the high near-frontier mass region but 60.3\% in the low near-frontier mass region. These results suggest that $\widehat{LB}$, which relies only on skewness and variance, offers a robust alternative to point estimation of $E[u]$.

The lower bound remains stable because it does not require recovering the support width or detailed distributional shape near zero; it depends only on the second and third moments. Next, we examine the near-frontier mass constraint as regularization for mapping estimated second through fourth central moments to $E[u]$. Figure~\ref{fig:mse_heatmap} shows $\log_{10}(\mathrm{MSE})$ of $\widehat{E}[u]$ over the parameter grid for the unconstrained estimator and the two constrained estimators $(m_0,c)=(1,1)$ and $(m_0,c)=(1,0.5)$. Table~\ref{table:bias_mse} reports median bias, median MSE, and mean MSE by region.

For $n=25$, the unconstrained mapping performs reasonably over much of the grid but can generate very large errors in some regions, especially the unimodal left-skewed region. This instability reflects noisy estimation of the central moments. When the distribution is roughly symmetric and tightly concentrated at an interior mode, skewness is close to zero and the second and fourth moments constrain the shape of the distribution much more than its location relative to zero. Small sampling errors in the estimated moments can then be matched by fitted distributions with substantially different support widths $q$, translating into large dispersion in the implied mean. This is visible in Figure~\ref{fig:mse_heatmap}, where the unconstrained estimator has very high MSE in parts of the grid, and in Table~\ref{table:bias_mse}, where it has much larger mean MSE than median MSE in several regions.

\begin{figure}[!t]
	\centering
	
	\begin{subfigure}{0.32\textwidth}
		\centering
		\includegraphics[width=\linewidth,trim=6 3 6 6,clip]{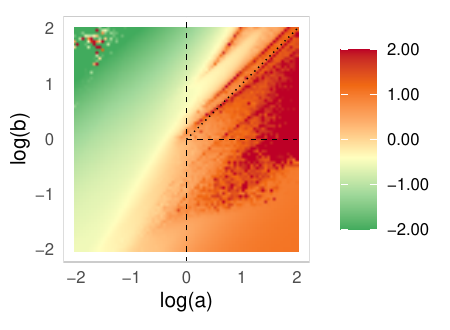}
		\caption{$n=25$, unconstrained}
	\end{subfigure}\hfill
	\begin{subfigure}{0.32\textwidth}
		\centering
		\includegraphics[width=\linewidth,trim=6 3 6 6,clip]{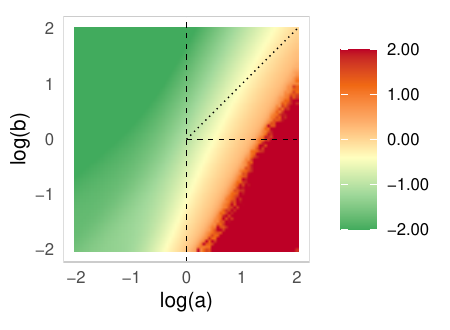}
		\caption{$n=250$, unconstrained}
	\end{subfigure}\hfill
	\begin{subfigure}{0.32\textwidth}
		\centering
		\includegraphics[width=\linewidth,trim=6 3 6 6,clip]{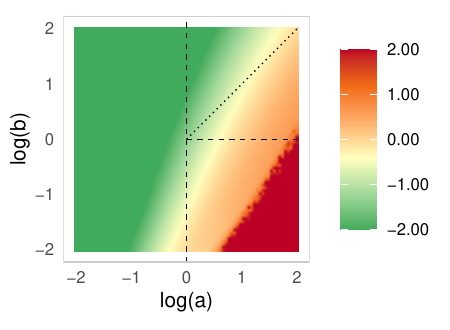}
		\caption{$n=2,500$, unconstrained}
	\end{subfigure}
	
	\vspace{0.3cm}
	
	\begin{subfigure}{0.32\textwidth}
		\centering
		\includegraphics[width=\linewidth,trim=6 3 6 6,clip]{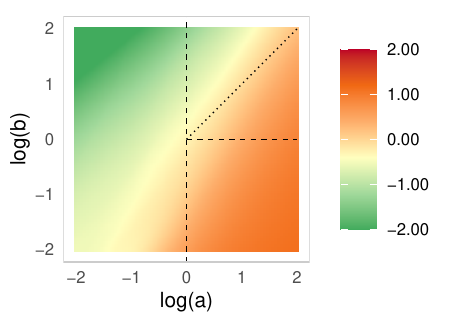}
		\caption{$n=25$, $(m_0,c)=(1,1)$}
	\end{subfigure}\hfill
	\begin{subfigure}{0.32\textwidth}
		\centering
		\includegraphics[width=\linewidth,trim=6 3 6 6,clip]{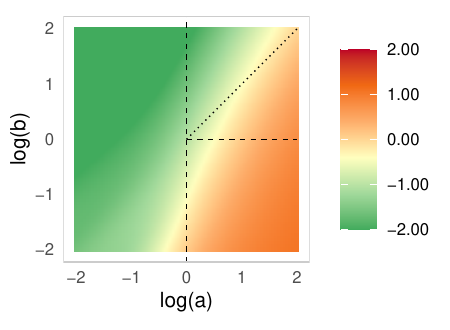}
		\caption{$n=250$, $(m_0,c)=(1,1)$}
	\end{subfigure}\hfill
	\begin{subfigure}{0.32\textwidth}
		\centering
		\includegraphics[width=\linewidth,trim=6 3 6 6,clip]{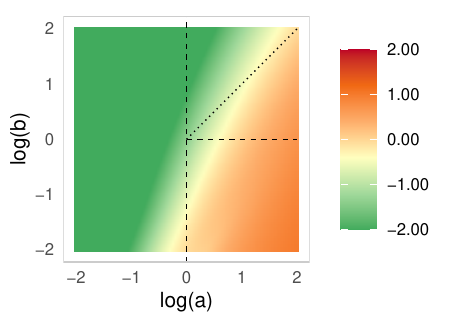}
		\caption{$n=2,500$, $(m_0,c)=(1,1)$}
	\end{subfigure}
	
	\vspace{0.3cm}
	
	\begin{subfigure}{0.32\textwidth}
		\centering
		\includegraphics[width=\linewidth,trim=6 3 6 6,clip]{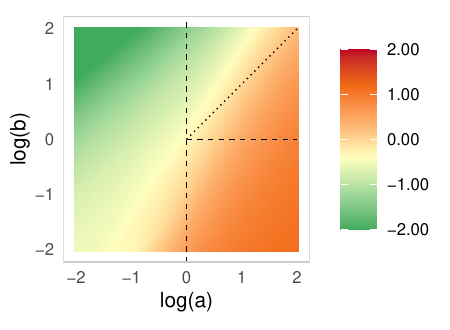}
		\caption{$n=25$, $(m_0,c)=(1,0.5)$}
	\end{subfigure}\hfill
	\begin{subfigure}{0.32\textwidth}
		\centering
		\includegraphics[width=\linewidth,trim=6 3 6 6,clip]{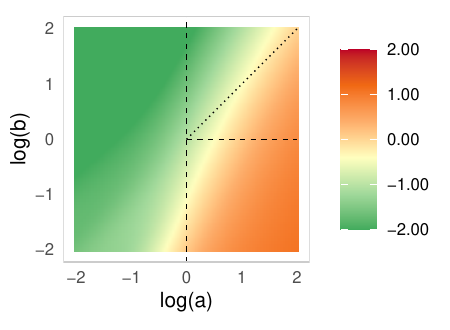}
		\caption{$n=250$, $(m_0,c)=(1,0.5)$}
	\end{subfigure}\hfill
	\begin{subfigure}{0.32\textwidth}
		\centering
		\includegraphics[width=\linewidth,trim=6 3 6 6,clip]{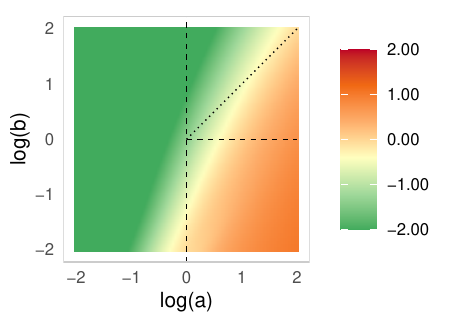}
		\caption{$n=2,500$, $(m_0,c)=(1,0.5)$}
	\end{subfigure}
	
	\caption{
		Heatmaps of $\log_{10}(\mathrm{MSE})$ of $\widehat{E}[u]$ over the parameter grid $(\log a,\log b)$. Columns correspond to sample sizes $n\in\{25,250,2\,500\}$. Rows correspond to the unconstrained estimator $(c=\infty)$ and the constrained estimators with $(m_0,c)=(1,1)$ and $(m_0,c)=(1,0.5)$. At each grid point, MSE is computed across the 150 Monte Carlo replications. The color scale is fixed at $[-2,2]$, with values outside this range clipped.
	}
	\label{fig:mse_heatmap}
\end{figure}


\begin{table}[!htb]
	\centering
	\footnotesize
	\setlength{\tabcolsep}{3.0pt}
	\begin{threeparttable}
		\caption{Bias and MSE by Region}
			\label{table:bias_mse}
		\begin{tabular}{@{}ll cccc cccc cccc@{}}
			\toprule
				&& \multicolumn{4}{c}{Median Bias} & \multicolumn{4}{c}{Median MSE} & \multicolumn{4}{c}{Mean MSE} \\
			\cmidrule(lr){3-6} \cmidrule(lr){7-10} \cmidrule(lr){11-14}
			$m_0$ & $c$ & High-NFM & Uni-R & Uni-L & Low-NFM & High-NFM & Uni-R & Uni-L & Low-NFM & High-NFM & Uni-R & Uni-L & Low-NFM \\
					\midrule
		\multicolumn{14}{@{}l}{\textit{Panel A: ($n=25$)}} \\
		- & $\infty$ & -0.094 & -0.203 & -0.027 & -1.650 & 0.174 & 0.953 & 41 & 10 & $>\!10^{34}$ & 2986 & $>\!10^{5}$ & $>\!10^{6}$ \\ 
		1 & $1.0$ & -0.094 & -0.377 & -1.384 & -2.371 & 0.165 & 0.363 & 2.107 & 5.857 & 0.273 & 0.426 & 2.559 & 6.026 \\ 
		1 & $0.5$ & -0.094 & -0.424 & -1.468 & -2.421 & 0.161 & 0.338 & 2.284 & 6.048 & 0.262 & 0.423 & 2.716 & 6.181 \\ 
		\midrule
		\multicolumn{14}{@{}l}{\textit{Panel B: ($n=250$)}} \\
		- & $\infty$ & -0.039 & -0.240 & -0.865 & 0.908 & 0.014 & 0.094 & 1.166 & 222 & 0.045 & 0.157 & $>\!10^{6}$ & $>\!10^{133}$ \\ 
		1 & $1.0$ & -0.039 & -0.241 & -1.046 & -1.779 & 0.014 & 0.093 & 1.174 & 3.350 & 0.045 & 0.156 & 1.538 & 3.987 \\ 
		1 & $0.5$ & -0.039 & -0.242 & -1.093 & -1.871 & 0.014 & 0.093 & 1.236 & 3.609 & 0.045 & 0.159 & 1.641 & 4.182 \\ 
		\midrule
		\multicolumn{14}{@{}l}{\textit{Panel C: ($n=2{,}500$)}} \\
		- & $\infty$ & -0.013 & -0.214 & -1.024 & -0.720 & 0.002 & 0.049 & 1.065 & 2.419 & 0.021 & 0.114 & 1.659 & $>\!10^{191}$ \\ 
		1 & $1.0$ & -0.013 & -0.214 & -1.026 & -1.483 & 0.002 & 0.049 & 1.065 & 2.294 & 0.021 & 0.114 & 1.369 & 3.005 \\ 
		1 & $0.5$ & -0.013 & -0.214 & -1.027 & -1.567 & 0.002 & 0.049 & 1.067 & 2.512 & 0.021 & 0.114 & 1.384 & 3.278 \\ 
		\bottomrule
		\end{tabular}
		\begin{tablenotes}[para,flushleft]
			\hspace*{-0.4em} Notes:
	Regions are defined by shape parameters $(a,b)$:
High-NFM, high near-frontier mass ($a<1$);
Uni-R, unimodal right-skewed ($1 \le a \le b$);
Uni-L, unimodal left-skewed ($1 \le b < a$);
Low-NFM, low near-frontier mass ($a \ge 1,\, b<1$).
For each grid point, bias and MSE are first computed across the 150 Monte Carlo replications.
The reported entries then aggregate these grid-point statistics across all grid points in the region.
Median Bias is the median across grid points of the bias of $\widehat{E}[u]$.
Median MSE is the median across grid points of the MSE of $\widehat{E}[u]$.
Mean MSE is the mean across grid points of the MSE of $\widehat{E}[u]$.
		\end{tablenotes}
	\end{threeparttable}
\end{table}

As the sample size increases and the moments are estimated more precisely, these failures become concentrated in regions where the moment-to-mean mapping is ill-conditioned: specifically, where the distribution of $u$ places almost all probability far from the frontier. The mean $E[u]$ depends on both the shape of the distribution and the support width $q$. When most probability is concentrated away from zero, the second through fourth moments mainly reflect the shape of the upper tail and provide limited information about $q$. The parametric beta family does not resolve this ambiguity because the induced scarcity mechanism implies that the realized distribution of $u$ is not exactly beta. As a result, the mapping can match the observed moments by shifting the fitted distribution along its support, increasing $q$ while preserving its shape, with little penalty, since there are no observations near the frontier to discipline the fit. In Figure~\ref{fig:mse_heatmap}, this appears as regions where the unconstrained estimator has very large MSE even at $n=2{,}500$.

The unconstrained estimator has the smallest bias in magnitude across regions and sample sizes (Table~\ref{table:bias_mse}, Median Bias column; see also Table~\ref{table:bind_absbias}, Median $|\text{Bias}|$ column).  However, its MSE can be orders of magnitude larger. Regularization mitigates these failures by requiring the fitted distribution to place a minimal probability mass in a neighborhood of zero. The constraint binds mainly where near-frontier mass is low and rarely where it is high (see Table~\ref{table:bind_absbias} in Appendix~\ref{ap:tables_figures}). For example, under $(m_0,c)=(1,1)$ at $n=2{,}500$, the bind share ranges from 0 in high near-frontier mass to 0.466 in low near-frontier mass. In regions where the constraint rarely binds, the constrained and unconstrained estimators have nearly identical median bias and median MSE. Where it binds, it reduces the extreme failures of the unconstrained mapping and sharply lowers mean MSE. At the same time, the constrained estimators generally exhibit more negative bias, especially in regions with little near-frontier mass. Thus the near-frontier mass constraint stabilizes point estimation at the cost of bias in regions where the true near-frontier mass falls short of the imposed threshold. The lower bound, which does not require recovering support width or distributional shape near zero, remains stable even where point estimation is weakly informative.



\section{Empirical Application to Production}\label{se:ap}

Consider the log-production model
\begin{align}
	y &= g(k -\omega_k, l -\omega_l) - \omega_y + v \label{eq:pf}
	= g(k, l) - u + v,
\end{align}
where $y$ is log output, $k$ is log capital,  $l$ is log labor, and $g(\cdot)$ is strictly increasing in each input.\footnote{If $g(\cdot)$ is CES, then one of $\omega_k$, $\omega_l$ or $\omega_y$ is redundant.} The unobservables $\omega_k$ and $\omega_l$ represent inefficient applications of capital and labor, respectively, while $\omega_y$ is a Hicks-neutral shock.\footnote{The Hicks-neutral shock $\omega_y$ can represent the effect of an unobserved input.} These three inefficiencies all have support $[0, \infty)$.  The total inefficiency in production is $u = g(k, l ) - [g(k - \omega_k, l - \omega_l) - \omega_y] \geq 0$, with $u=0$ if and only if the firm is efficient. The random error $v$ represents measurement error in $y$, including unpredictable random productivity shocks, and is assumed  to be mean independent of $(k,l)$.
	
Endogenous inputs are a pervasive challenge in estimating production functions. The most widely used approaches for estimating the conditional mean production function combine panel data with control functions, using a proxy variable (e.g., investment or intermediate inputs) that is monotone in unobserved productivity and exploiting timing and information-set assumptions to control for simultaneity \citep[see, e.g.,][]{olleypakes,levinsohn2003estimating,ackerberg2015identification}. In contrast, our frontier approach replaces proxy-variable and timing assumptions with bounded-above productivity ($u\geq 0$) and assignment at the frontier ($0 \in \supp(u)$) assumptions, which together identify the frontier relationship without instruments or proxies. Classical SFA methods also use the idea of a frontier, but typically proceed by imposing parametric structure on the frontier and distributional assumptions on inefficiency and random error; accommodating endogeneity is generally thought to require additional restrictions (e.g., controls or instruments).

To illustrate our approach, we use plant-level panel data and treat inefficiencies as time-invariant, as in \eqref{eq:model_fe1}--\eqref{eq:model_fe2}:
\begin{align*}
y_{it}=g(x_{it})-u_i+v_{it}, \qquad x_{it}:=(k_{it},l_{it}), \quad u_i\ge 0, \quad  t=1,\ldots,T_i, \quad i=1,\ldots,n.
\end{align*}
 In model \eqref{eq:pf} described above, inefficiency enters as $g(k - \omega_k, l - \omega_l)$ so under a flexible $g(\cdot)$, the implied output loss generally depends on contemporaneous $(k,l)$. This would make the inefficiency component time-varying. Accordingly, in the application we restrict inefficiency to be Hicks-neutral ($\omega_k=\omega_l=0$).  Following \citet{Mundlak1978}, we further restrict the dependence between inputs and unobservables in Assumption \ref{asn:v} to operate through $\bar x_i = T_i^{-1}\sum_{t=1}^{T_i} x_{it}$ (the time-averaged inputs): (i) $E[v_{it}\mid x_{it},\bar x_i]=0$; (ii) $	\mu_{k,v}(x_{it},\bar x_i) = \mu_{k,v}(\bar x_i)$, for $k\in\{2,3,4\}$; (iii) $u_i \mid (x_{i1},\ldots,x_{iT_i}) \overset{d}{=} u_i \mid \bar x_i $; (iv) $u_i \indep (v_{i1},\ldots,v_{iT_i}) \mid \bar x_i.$

These restrictions reduce the conditioning set from the full input history $(x_{i1},\ldots,x_{iT_i})$ to the low-dimensional summary $\bar x_i$, making nonparametric estimation of the conditional moments feasible. In principle, the identification strategy applies with conditioning on the full history; the Mundlak restriction is a practical dimension reduction.  In the repeated measurement model $y_{it}=g(x_i)-u_i+v_{it}$ of Appendix~\ref{ap:multi-meas}, neither restriction is needed: because inputs do not vary within plant,  inefficiency $g(k_i,l_i)-g(k_{i}-\omega_{k},l_{i}-\omega_{l})+\omega_{y}$ is time-invariant without imposing $\omega_k=\omega_l=0$, and $\bar x_i = x_i$ by construction, so the Mundlak assumption imposes no additional restriction beyond conditioning on $x_i$ itself.

Our data are plant-level observations from the Colombian manufacturing census for 1981--1991. We focus on the food products industry (ISIC 311), using the \texttt{colombian} dataset distributed with the \texttt{gnrprod} \texttt{R} package.\footnote{See \citet{gandhi2020identification} for discussion of this dataset and its use in production-function applications.} This application is intended to illustrate the method in a standard production setting, rather than to deliver a definitive structural account of Colombian manufacturing.

We set output to log real gross output ($y=\texttt{RGO}$), capital to log real capital stock ($k=\texttt{K}$), and labor to log labor input measured in employee-years ($l=\texttt{L}$), where all variables are expressed in logs and output and capital are deflated to constant prices. After excluding observations with fewer than 10 employee-years and then restricting the sample to plants observed for at least 8 years to ensure sufficient within-plant variation, we obtain an unbalanced panel of $n=408$ plants and $4,306$ plant-year observations, with an average duration of $10.6$ years per plant; $24$ plants are observed for $8$ years, $29$ for $9$ years, $52$ for $10$ years, and $303$ plants are present for the full $11$ years. In the sample, the mean log output is $9.21$ (s.d. $1.73$), with right-skewness of $0.20$ and kurtosis of $5.28$. Mean log labor is $4.05$ (s.d. $1.21$) and mean log capital is $7.37$ (s.d. $1.92$). 

For mean regression in such a panel setting, a fixed effects (within) estimator would be natural, since differencing removes the time-invariant inefficiency term and thereby eliminates bias from correlation between inefficiency and inputs. Our interest, however, lies in the frontier relationship. Under assignment at the frontier, correlation between inefficiency and inputs is not a concern for identification of the frontier. Fixed effects therefore addresses a problem that is absent under our identifying assumptions, while discarding the between-plant differences that are precisely the most informative about which plants operate closest to the frontier.

In production applications, fixed effects estimation also has a long record of producing implausible estimates, particularly for the capital coefficient, which can be attenuated toward zero \citep[e.g.,][]{griliches1998production}. One reason for this could be the limited within-plant variation in capital relative to between-plant variation; in our data, within-plant variance accounts for about 5\% of total capital variance and 9\% for labor. Moreover, consistency of the fixed effects estimator requires a strict exogeneity condition that productivity shocks are mean independent of the entire history of inputs. This rules out feedback from past productivity shocks to future input choices. By contrast, we maintain a weaker mean independence restriction that conditions only on contemporaneous inputs and their time averages, allowing for more general forms of dynamic adjustment in input choices.

We implement the three-stage estimator developed in Section~\ref{se:est} with the Mundlak assumptions above as follows. First, we nonparametrically estimate the conditional expectation $E[y_{it}\mid x_{it},\bar x_i]$ using a tensor-product thin-plate regression spline basis. To separate within-plant and between-plant variation, the specification includes main effects and interactions in the time-averaged inputs $\bar{x}_i$, as well as main effects and interactions in the demeaned inputs $x_{it}-\bar{x}_i$. We also include interactions between the within and between components. This provides a flexible first-stage approximation, while nesting the additive structure implied by the maintained model. We set each marginal basis dimension on the order of $n^{1/2},$ where $n$ is the number of plant-year observations, yielding a total of 132 basis functions across terms, and estimate the resulting high-dimensional linear model by ridge regression. The fitted values yield centered residuals $\hat \eps_{it}$, which we decompose into within-plant residuals $\hat \eps_{it}^w$ and plant-level residuals $\bar{\hat \eps}_i$ (the sample analogs of \eqref{eq:within_id} and \eqref{eq:between_id}).

Second, we estimate conditional central moments using the within-between decomposition. Using empirical central moments of $\hat \eps_{it}^w$, we obtain plant-level error central moments $\hat\mu_{k,v}(\bar x_i)$ using \eqref{eq:mu2vhat}--\eqref{eq:mu4vhat}, and smooth these as functions of $\bar x_i$ using ridge regression on the spline basis in $\bar x_i$. We then construct adjusted plant-level powers $\hat u_i^k$ using \eqref{eq:ui2}--\eqref{eq:ui4} in Appendix~\ref{ap:estdetails_sample} and obtain smoothed estimates of $\hat\mu_{k,u}(\bar x_i)$ for $k\in\{2,3,4\}$ by estimating these adjusted powers as functions of $\bar x_i$ by ridge regression.

Finally, we estimate the inefficiency distribution by matching the estimated central moments to model-implied central moments. For each plant, we consider two parametric families for $u\mid \bar x$: the scaled beta (as in our simulations) and truncated normal distributions (a generalization of the half-normal used in the original SFA literature, and often used in that literature now). We compute $\hat\theta(\bar x)$ by minimizing the distance between estimated and model-implied central moments subject to the near-frontier mass restriction in \eqref{eq:mm_constrained}. We set the constraint parameters to $(m_0,c)=(1,0.5)$, which ensures that the expected number of plants within $0.5$ standard deviations of the frontier is at least one.  In the input-dependent case, $\widehat n_{\mathrm{eff}}(\bar x)$ is computed using kernel weights (bandwidth $h=0.20$), yielding an average effective sample size of about $30$, and a standard deviation of $15$ across plants. Consequently, the required probability mass $p_0(\bar x) = 1/\widehat n_{\mathrm{eff}}(\bar x)$ adapts to the local data density, averaging approximately $0.03$, which provides stronger regularization in regions where the effective sample size is small.

We first consider the case where the distributions of $u_{i}$ and $v_{it}$ are independent of inputs. This corresponds to the Corrected Ordinary Least Squares (COLS) estimator, where the conditional central moments reduce to constants (and $g(x)=E[y\mid x]+E[u]$). Denote the estimate of the negative plant-level mean centered residual by $\hat{\eta}_i = -\bar{\hat{\eps}}_i = \reallywidehat{(u_i - E[u]) - \bar{v}_i}$ (an estimate of centered inefficiency contaminated by plant-averaged error). Figure~\ref{fig:density_zoom} plots the kernel density of the shifted estimate $\hat{\eta}_i-\min_j(\hat{\eta}_j)$, which sets the minimum value to zero to visualize the support relative to zero.

We estimate the central moments of the inefficiency $u$ by subtracting the estimated error central moments from the central moments of the between residuals (see \eqref{eq:mu2vhat_pool}--\eqref{eq:mu4uhat_pool} in~Appendix \ref{ap:homo}). The resulting central moment estimates are $\widehat\mu_{2,u} \approx 0.59$, $\widehat\mu_{3,u} \approx 0$, and $\widehat\mu_{4,u} \approx 1.09$. These imply a skewness of $0$ and kurtosis of $3.13$. The small skewness and kurtosis close to 3 are consistent with an approximately symmetric distribution.

We fit a scaled beta distribution $q\cdot\mathrm{Beta}(a,b)$ and a truncated normal distribution $TN(\mu,\sigma^2)$ on $[0,\infty)$ to these estimated central moments. The unconstrained fitted densities are shown as dashed lines in Figure~\ref{fig:density_zoom}. These fitted densities reproduce the overall shape of the empirical kernel density (shifted so minimum value is zero) but imply implausibly large mean inefficiency: $\widehat E[u]\approx 20.87$ (beta) and $\widehat E[u]\approx 4.80$ (truncated normal). 

This motivates the near-frontier mass restriction \eqref{eq:mm_constrained} as a regularizer. The solid lines in the figures impose $(m_0,c)=(1,0.5)$, ensuring that the expected number of plants within approximately $ \widehat\sigma_u/2 \approx 0.38$ log points of the frontier is at least one. Equivalently, with $\widehat n_{\mathrm{eff}}=n=408$ this requires at least $p_0=m_0/n=1/408\approx 0.002$ probability mass within $0.5$ standard deviations of zero, ruling out fits that place essentially no mass near the frontier. This results in $\widehat E[u]\approx 2.30$ (beta) and $\widehat E[u]\approx 2.53$ (truncated normal).

\begin{figure}[!t]
	\centering
	\includegraphics[width=0.55\textwidth]{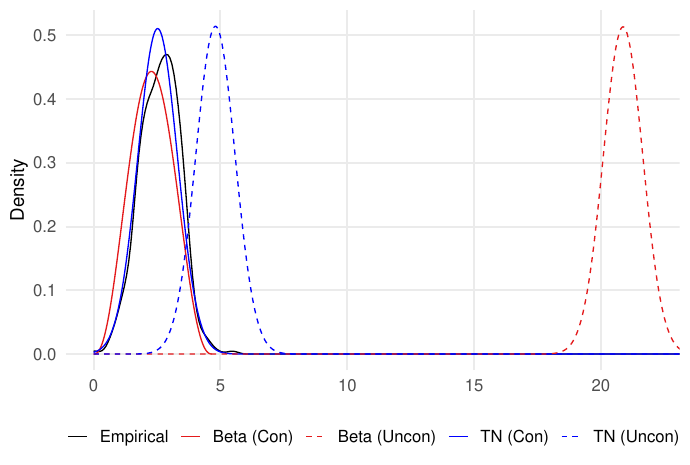}
	\caption{
			 Empirical and fitted densities of estimated inefficiency. The solid black line shows the kernel density (minimum shifted to zero) of $\hat{\eta}_i-\min_j(\hat{\eta}_j)$. Dashed lines show unconstrained fits and solid lines constrained fits: red for beta and blue for truncated normal.}
	\label{fig:density_zoom}
\end{figure}

We now turn to estimation allowing the distributions of $u$ and $v$ to depend on inputs through the Mundlak controls $\bar x_i$. Following the within--between decomposition described above, we obtain plant-level estimates of the conditional central moments of both components. Across plants, the estimated error second central moment has mean $0.15$ (s.d.\ $0.08$), the third central moment has mean $0.02$ (s.d.\ $0.01$), and the fourth central moment has mean $0.21$ (s.d.\ $0.11$). The estimated inefficiency second central moment has mean $0.61$ (s.d.\ $0.24$), the third central moment has mean $0.00$ (s.d.\ $0.01$), and the fourth central moment has mean $1.08$ (s.d.\ $0.29$).

We report the coefficients from regressions of estimated mean inefficiency on plant-level mean inputs (Table~\ref{table:estimates_combined}, Panel~A). The coefficient on labor is $-0.39$ (beta) and $-0.39$ (truncated normal), and the coefficient on capital is $0.20$ (beta) and $0.14$ (truncated normal). The positive coefficient on capital suggests that estimated mean inefficiency is correlated with capital, in contrast with the control-function approach of \cite{olleypakes}, which assumes that the proxy variable (e.g., investment) is strictly increasing in productivity conditional on capital. Under this strict monotonicity, persistently more productive plants will have greater capital.

\begin{table}[!t]
	\centering
		\footnotesize
\setlength{\tabcolsep}{2.5pt}
	\caption{Production function and inefficiency regressions}
	\label{table:estimates_combined}
	\begin{threeparttable}
		\begin{tabular}{l c c c}
			\toprule
			& Labor ($l$) & Capital ($k$) & E/C Ratio \\
			\midrule
			\multicolumn{4}{l}{\textit{Panel A: Inefficiency}} \\
			MM (beta) & -0.39 & 0.20 & \\
			& (0.07) & (0.05) & \\
			MM (truncated normal) & -0.39 & 0.14 & \\
			& (0.05) & (0.03) & \\
			\midrule
			\multicolumn{4}{l}{\textit{Panel B: Production function}} \\
			OLS (within) & 0.27 & 0.20 & 4.62 \\
			& (0.02) & (0.01) & \\
			Levinsohn-Petrin & 0.12 & 0.15 & 2.68 \\
			& (0.02) & (0.05) & \\
			SFA (truncated normal) & 0.32 & 0.25 & 4.34 \\
			& (0.01) & (0.01) & \\
			\addlinespace
			MM (beta) & 0.21 & 0.61 & 1.15 \\
			& (0.05) & (0.04) & \\
			MM (truncated normal) & 0.20 & 0.56 & 1.22 \\
			& (0.04) & (0.03) & \\
			\bottomrule
		\end{tabular}
	\end{threeparttable}
	
	\vspace{2pt}
	\begin{minipage}{\textwidth}
		\footnotesize
		\textit{Notes:} Panel A reports regressions of estimated mean inefficiency on plant-level mean inputs. Panel B reports production function estimates. OLS and Levinsohn-Petrin estimate the mean production function, while SFA and Method of Moments (MM) estimate the frontier production function. Standard errors (in parentheses) for MM and Levinsohn-Petrin are calculated using 100 plant-level bootstrap samples, while those for OLS and SFA use analytic formulas. The E/C Ratio column reports the relative elasticity--cost share ratio, calculated as $0.05 \times (\beta_l / \beta_k) / (\bar{L} / \bar{K})$, assuming a capital rental rate of 5\%.
	\end{minipage}
\end{table}

We next compare production-function elasticity estimates across procedures (Table~\ref{table:estimates_combined}, Panel~B). For the mean production function, OLS (within) regresses demeaned output on demeaned inputs, yielding elasticities of $0.27$ on labor and $0.20$ on capital, while the Levinsohn--Petrin estimator yields $0.12$ on labor and $0.15$ on capital. For the frontier, the parametric SFA model with truncated normal time-invariant inefficiency yields elasticities of $0.32$ on labor and $0.25$ on capital.\footnote{Under the OLS and SFA assumptions that inefficiency and error are mean independent of inputs ($E[u_i\mid x_{i1},\ldots,x_{iT}]=E[u_i]$ and $E[v_{it}\mid x_{i1},\ldots,x_{iT}]=0$), we have $E[y_{it}\mid x_{i1},\ldots,x_{iT}]=g(x_{it})-E[u_i]$, so the frontier differs from the conditional mean by a constant shift, and one would expect identical slope coefficients across mean and frontier regressions. The estimates differ in Table~\ref{table:estimates_combined} because OLS identifies $\beta$ from within-plant variation only, whereas the panel SFA estimator exploits total variation and estimates parameters by maximum likelihood. Using comparable estimators, cross-sectional SFA yields slope coefficients and intercepts nearly identical to cross-sectional OLS, implying negligible estimated inefficiency. This is consistent with diagnostic tests suggesting that standard SFA is unable to distinguish inefficiency from error in this setting.} Using our moment-based approach, we estimate $\widehat g(\bar x_i)=\widehat E[y\mid \bar x_i]+\widehat E[u_i\mid \bar x_i]$ and obtain frontier elasticities by regressing $\widehat g$ on labor and capital. The estimated frontier elasticities are $0.21$ (beta) and $0.20$ (truncated normal) for labor, and $0.61$ (beta) and $0.56$ (truncated normal) for capital. The MM frontier estimates imply decreasing returns to scale, with coefficient sums of 0.82 (beta) and 0.76 (truncated normal), which are substantially larger than the corresponding OLS and Levinsohn–Petrin estimates. The elasticity--cost share ratio (E/C Ratio) compares the relative labor--capital elasticity to the corresponding relative input usage; under neoclassical price-taking behavior with unrestricted input choice, this ratio equals one. In our estimates, the implied ratio ranges from $1.15$ (MM (beta)) and $1.22$ (MM (truncated normal)) to $4.62$ (OLS).
 
Finally, we report the unconditional lower bound and mean inefficiency estimates. The estimated lower bound on mean inefficiency is $0.76$. Our MM estimates of mean inefficiency (averaged over plants) are $1.43$ (beta) and $1.67$ (truncated normal). For comparison, the SFA model (assuming time-invariant inefficiency) yields a higher mean inefficiency of $2.75$, while COLS yields $2.30$ (beta) and $2.53$ (truncated normal). Mean inefficiency estimates are sensitive to whether one allows $(u_i,v_{it})$ to depend on inputs. In our data, the MM estimates that allow such dependence are lower than input-independent benchmarks, consistent with the possibility that input-independent methods overestimate inefficiency by attributing heterogeneity or input--error correlations to technical inefficiency.

To visualize the estimated structural frontier $g(x)$ against the data, Figure~\ref{fig:marginal_slice} plots slices of the production surface. Each slice shows output against one input, restricting the other input to lie within a bandwidth around its sample median. The plotted curves represent smoothed estimates of the frontier, conditional mean, and 95\% quantile regression. For example, Panel~(a) plots log output against log capital for plants with log labor near the sample median. The slices show a broadly concave frontier consistent with diminishing marginal returns to capital. The estimated frontier lies near the 95\% conditional percentile and is slightly more concave than the mean estimates.

\enlargethispage{\baselineskip}
\begin{figure}[!htbp]
	\centering
	\begin{subfigure}{0.48\textwidth}
		\centering
		\includegraphics[width=\linewidth,trim=2 3 6 6,clip]{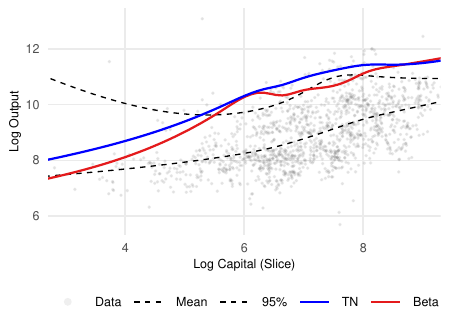}
		\caption{Slice: Capital ($l=3.84 \approx \text{median}$)}
	\end{subfigure}\hfill
	\begin{subfigure}{0.48\textwidth}
		\centering
		\includegraphics[width=\linewidth,trim=2 3 6 6,clip]{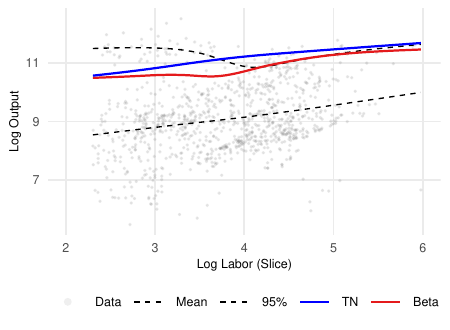}
		\caption{Slice: Labor ($k=7.40 \approx \text{median}$)}
	\end{subfigure}
	\caption{ 
		Panel (a) plots output against capital near the median labor input; Panel (b) plots output against labor near the median capital input. Red (beta) and blue (truncated normal) curves show the smoothed frontier; black dashed curves show the conditional mean and 95th percentile; gray points are raw data.
	}	
	\label{fig:marginal_slice}
\end{figure}

\FloatBarrier
\section{Conclusion}  \label{se:con}

This paper identifies the frontier structural function from the supremum of the outcome given inputs, assuming a nonnegative deviation and assignment at the frontier (that zero lies in the deviation’s support given inputs). This identification holds even when inputs are endogenous, thereby obviating the need for instrumental variables. We then allow for random error that is mean independent of inputs. Using only observed inputs and outcomes, we develop estimators of the frontier structural function and mean deviation based on conditional central moments.  We regularize the moment-matching step by imposing a near-frontier mass constraint that the fitted deviation distribution retain a minimum amount of probability mass near the frontier. We also derive a skewness-based lower bound on the mean deviation that is robust to scarcity of data near the frontier. Monte Carlo simulations show that the near-frontier mass constraint substantially improves finite-sample performance by reducing mean squared error, while the skewness-based lower bound remains accurate even when near-frontier observations are scarce.

In an application to the Colombian food products industry, estimated mean inefficiency is correlated with inputs. There are a number of possible explanations for this pattern, including that plants choose inputs based on productivity, that firms face different input prices according to productivity, or that inefficiency in input choice is correlated with inefficiency in input application; distinguishing among these explanations would require additional information or stronger structure. The regularized estimator yields fitted distributions that capture the overall shape of the empirical distribution while satisfying the near-frontier mass constraint, and the skewness-based lower bound provides an optimization-free alternative to these point estimates.

\begingroup
\setstretch{1.30}
\small
\bibliography{biblio}
\endgroup

\pagebreak


\appendix

\noindent {\LARGE \bf Appendix}


\section{Assignment at the Frontier and Near-Frontier Mass Constraint}\label{ap:AAF}

Consider sequences $\delta_n \downarrow 0$ and $p_n>0$ and the restrictions $F_{u|x}(\delta_n) \ge p_n$.

\begin{prop}\label{prop:seq_to_support}
	If there exist sequences $\delta_n \downarrow 0$ and $p_n>0$ such that
	$F_{u|x}(\delta_n)\ge p_n$ for all $n$, then $0\in \Supp(u\mid x)$.
\end{prop}
\vspace{-0.2cm}
\begin{proof}
	Fix any $\varepsilon>0$. Since $\delta_n \downarrow 0$, there exists an integer $n(\varepsilon)$ such that
	$\delta_{n(\varepsilon)}<\varepsilon$. By monotonicity of $F_{u|x}$, we have
	$
	F_{u|x}(\varepsilon)\ \ge\ F_{u|x}\!\big(\delta_{n(\varepsilon)}\big)\ \ge\ p_{n(\varepsilon)}\ >\ 0.
	$
	Since this holds for all $\varepsilon>0$, it follows that $0\in \Supp(u\mid x)$.
\end{proof}
\vspace{-0.2cm}
\begin{prop} \label{prop:support_to_seq}
	If $0\in \Supp(u\mid x)$, then for any sequence $\delta_n\downarrow 0$,
	the choice $p_n := F_{u|x}(\delta_n)$ satisfies $p_n>0$ for all $n$ and
	$F_{u|x}(\delta_n)= p_n$ holds. Moreover $p_n \downarrow F_{u|x}(0)$.
\end{prop}
\vspace{-0.2cm}
\begin{proof}
	By the support assumption, $p_n=F_{u|x}(\delta_n)>0$ for all $n$ since $\delta_n>0$.
	Since $\delta_n\downarrow 0$ and $F_{u|x}$ is nondecreasing and right-continuous,
	$p_n=F_{u|x}(\delta_n)\downarrow F_{u|x}(0)$.
\end{proof}
Proposition~\ref{prop:seq_to_support} guarantees existence of shrinking-neighborhood inequalities, but it does not guarantee that an arbitrary pre-specified pair $(\delta_n\downarrow 0,p_n>0)$ will satisfy $F_{u|x}(\delta_n)\ge p_n$ for large $n$. Proposition~\ref{prop:support_to_seq} only asserts $F_{u|x}(\varepsilon)>0$ for every $\varepsilon>0$; it does not quantify how small $F_{u|x}(\varepsilon)$ may be as $\varepsilon\downarrow 0$.
For example, define $F(t)=\exp(-1/t)$ for $t>0$ and $F(t)=0$ for $t\le 0$. This satisfies $F(\varepsilon)>0$ for every $\varepsilon>0$, so $0\in \Supp(u)$, but the sequences $\delta_n=1/n^2$ and $p_n=1/n$ violate $F(\delta_n)\ge p_n$ for large $n$ since $F(\delta_n)=e^{-n^2}\ll 1/n=p_n$. In estimation, we restrict attention to parametric families with polynomial tails.

\section{The Primary Hankel Matrix} \label{ap:primaryhankel}

Define the central moments $\mu_j=E[(u - E[u])^j]$ and  Hankel matrix of central moments
{\small\begin{align*}
		\widetilde{\Delta}_n^{(0)} := \begin{pmatrix}
			1 & 0 & \mu_2 & \cdots & \mu_n\\
			0 & \mu_2 & \mu_3 & \cdots & \mu_{n+1}\\
			\mu_2 & \mu_3 & \mu_4 & \cdots & \mu_{n+2}\\
			\vdots & \vdots & \vdots & \ddots & \vdots\\
			\mu_n & \mu_{n+1} & \mu_{n+2} & \cdots & \mu_{2n}
		\end{pmatrix}.
	\end{align*}
}We show that $\det(\Delta_n^{(0)}) = \det(\widetilde{\Delta}_n^{(0)})$ by writing $(u - m_1)^j$ via the binomial theorem and noting that this results in a linear transformation between raw and central moments preserving the determinant.  
Specifically, define 
$
U = (1,u,u^2,\ldots,u^n)^\top$ and $
W = \bigl(1,(u - m_1),(u - m_1)^2,\ldots,(u - m_1)^n\bigr)^\top.
$
Then $\Delta_n^{(0)} = E[UU^\top]$ and $\widetilde{\Delta}_n^{(0)} = E[WW^\top]$.  Because $(u - m_1)^k = \sum_{j=0}^k \binom{k}{j}(-m_1)^{k-j} u^j,$ we obtain $W = AU$, where $A$ is the nonsingular lower-triangular matrix
{\small
	\begin{align*}
		A=
		\begin{pmatrix}
			1 & 0 & 0 & \cdots & 0 \\
			\binom{1}{0}(-m_1)^1 & 1 & 0 & \cdots & 0 \\
			\binom{2}{0}(-m_1)^2 & \binom{2}{1}(-m_1)^1 & 1 & \cdots & 0 \\
			\vdots & \vdots & \vdots & \ddots & \vdots \\
			\binom{n}{0}(-m_1)^n & \binom{n}{1}(-m_1)^{n-1} & \binom{n}{2}(-m_1)^{n-2} & \cdots & 1
		\end{pmatrix}.
	\end{align*}
}Hence,
$
\widetilde{\Delta}_n^{(0)} 
=
E[WW^\top] 
= 
AE[UU^\top]A^\top 
=
A\Delta_n^{(0)}A^\top.
$
So
$\det(\widetilde{\Delta}_n^{(0)}) = \det(\Delta_n^{(0)})(\det A)^2 = \det(\Delta_n^{(0)})$.

The primary Hankel matrices $\Delta_n^{(0)}$, via their relationship to the central moment matrices $\widetilde{\Delta}_n^{(0)}$ and their nonnegative determinants, yield inequalities in central moments. For example,
\vspace{-0.05cm}
{\footnotesize \begin{align*}
		0 &\leq  \mu_2, \quad
		0 \leq \mu_2\mu_4 - \mu_3^2 - \mu_2^3, \quad
		0 \leq \mu_3^4 - \mu_3^2(3\mu_2\mu_4 +\mu_6) + \mu_2^2(\mu_4^2 + 2\mu_3\mu_5)   - \mu_4^3+ \mu_4(2\mu_3\mu_5+\mu_2\mu_6) -\mu_2\mu_5^2 -\mu_2^3\mu_6  .
	\end{align*}
}The first inequality states that the variance of $u$ is nonnegative, so $u$ is nondegenerate. The second gives a well-known relationship between kurtosis and skewness:
${\mu_4}/{ \mu_2^2} >( {\mu_3^2}/{ \mu_2^{3/2}})^2 +1.$

\section{Moment Identities and Estimation Details} \label{ap:estdetails}

This appendix collects the moment identities and sample formulas used in Section~\ref{se:est}. We work under the panel setup and notation introduced there. For more details on the derivations see \cite{BenMosheGenesoveMoments}. 

\subsection{Within Identities and Error Central Moments}

Using the conditional mutual independence of $v_{i1},\ldots,v_{iT}$ given the inputs (Assumption~\ref{asn:v}(iii)) and the definition of $\eps_{it}^w$, algebra
yields the within moment relationships
\begin{align}
	\mu_{2,\eps^w}(x) &= \frac{T-1}{T} \mu_{2,v}(x), 	 \label{eq:mu2v} \\
	\mu_{3,\eps^w}(x) &= \frac{(T-1)(T-2)}{T^2}\mu_{3,v}(x), 	\label{eq:mu3v}\\
	\mu_{4,\eps^w}(x) &= \frac{(T-1)(T^2-3T+3)}{T^3}\mu_{4,v}(x)
	+ 3\frac{(T-1)(2T-3)}{T^3}\bigl(\mu_{2,v}(x)\bigr)^2. \label{eq:mu4v}
\end{align}
These can be inverted to express $\mu_{k,v}(x)$, $k=2,3,4$, in terms of
$\mu_{k,\eps^w}(x)$.

\subsection{Between Identities and Deviation Central Moments}

Using the conditional independence of $u_i$ and $\{v_{it}\}_{t=1}^T$ given the inputs (Assumption~\ref{asn:v}(ii)), the between residual $\bar \eps_i = -(u_i-E[u_i \mid  x) + \bar v_i$ yields
\begin{align}
	\mu_{2,\bar \eps}(x)
	&= \mu_{2,u}(x) + \frac{\mu_{2,v}(x)}{T}, \label{eq:mu2u}
	\\
	\mu_{3,\bar \eps}(x)
	&= -\mu_{3,u}(x) + \frac{\mu_{3,v}(x)}{T^2}, \label{eq:mu3u}
	\\
	\mu_{4,\bar \eps}(x)
	&= \mu_{4,u}(x)
	+ 6\,\mu_{2,u}(x)\,\frac{\mu_{2,v}(x)}{T}
	+ \frac{\mu_{4,v}(x)}{T^3}
	+ 3\frac{T-1}{T^3}\bigl(\mu_{2,v}(x)\bigr)^2. \label{eq:mu4u}
\end{align}
Substituting $\mu_{k,v}(x)$ from \eqref{eq:mu2v}--\eqref{eq:mu4v} into \eqref{eq:mu2u}--\eqref{eq:mu4u} expresses the deviation central moments $\mu_{k,u}(x)$ in terms of $\mu_{k,\eps^w}(x)$ and $\mu_{k,\bar \eps}(x)$.

\subsection{Sample Formulas for Conditional Central Moments}
\label{ap:estdetails_sample}

In Section~\ref{se:est}, the conditional mean $E[y \mid x]$ is estimated by a nonparametric regression, and residuals are obtained. 
Then, for each firm, sample analogs of the within identities \eqref{eq:mu2v}--\eqref{eq:mu4v} yield estimators of the conditional error central moments at $ x$:
\begin{align}
	\widehat\mu_{2,v}(x)
	&= \frac{1}{T_i-1}\sum_{t=1}^{T_i} (\hat \eps_{it}^w)^2, \label{eq:mu2vhat}\\
	\widehat\mu_{3,v}(x)
	&= \frac{T_i}{(T_i-1)(T_i-2)}\sum_{t=1}^{T_i} (\hat \eps_{it}^w)^3, \label{eq:mu3vhat}\\
	\widehat\mu_{4,v}(x)
	&= \frac{ (T_i^3-2T_i^2-3T_i+9)\sum_{t=1}^{T_i} (\hat \eps_{it}^w)^4 - 6(2T_i-3)\sum_{1\le t<t'\le T_i} (\hat \eps_{it}^w)^2(\hat \eps_{it'}^w)^2 }{T_i(T_i-1)(T_i-2)(T_i-3)}. \label{eq:mu4vhat}\\
	\widehat{\mu_{2,v}^2}(x) &= \frac{ 2(T_i^2-3T_i+3)\sum_{1\le t<t'\le T_i} (\hat \eps_{it}^w)^2(\hat \eps_{it'}^w)^2 - (2T_i-3)\sum_{t=1}^{T_i} (\hat \eps_{it}^w)^4 }{T_i(T_i-1)(T_i-2)(T_i-3)}. \label{eq:mu22vhat}
\end{align}

The between identities \eqref{eq:mu2u}--\eqref{eq:mu4u} relate the central moments of $\bar \eps_i$ to those of $u_i$ and $v_{it}$ and are used to subtract the
contribution of $\bar v_i$ from the between residual. We construct adjusted between--residual powers sequentially, substituting previously estimated moments as they become available. First, using the second and third between identities \eqref{eq:mu2u}--\eqref{eq:mu3u},
\begin{align}
	\widehat u_{i}^{2} &= (\bar{\hat \eps}_i)^2 - \frac{\widehat\mu_{2,v}(x)}{T_i}, \label{eq:ui2}\\
	\widehat u_i^{3} &= -(\bar{\hat \eps}_i)^3 + \frac{\widehat\mu_{3,v}(x)}{T_i^2}.\label{eq:ui3}
\end{align}
These are nonparametrically regressed on inputs to obtain $\widehat\mu_{2,u}(x)$ and $\widehat\mu_{3,u}(x)$. Next, using the fourth between identity \eqref{eq:mu4u} and substituting $\widehat\mu_{2,u}(x)$ from the previous step,
\begin{align}
	\widehat u_i^{4} &= (\bar{\hat \eps}_i)^4
	- 6\,\widehat\mu_{2,u}(x)\,\frac{\widehat\mu_{2,v}(x)}{T_i}
	- \frac{\widehat\mu_{4,v}(x)}{T_i^3}
	- 3\frac{T_i-1}{T_i^3}\,\widehat{\mu_{2,v}^2}(x). \label{eq:ui4}
\end{align}
This is regressed on inputs to obtain $\widehat\mu_{4,u}(x)$. By construction, the conditional expectations of these adjusted powers satisfy $E[\widehat u_i^{k} \mid x] = \mu_{k,u}(x)$ for $k=2,3,4$.

\section{Alternative Data Structures} \label{ap:datastruct}


Section~\ref{se:est} considers estimation with panel data with time-invariant deviations. This appendix discusses alternative data structures: repeated measurements, panel data in which the distributions of $u_i$ and $v_{it}$ are independent of $x$, and a cross-sectional setting in which identification relies on symmetry restrictions on $v$ rather than on within–between variation in panel data.

\subsection{Repeated Measurements} \label{ap:multi-meas}
Suppose that for each firm $i$ we observe $T$ repeated input values,
$
x_{i1} = \cdots = x_{iT} = x_i.
$
After subtracting the conditional mean,
$
\eps_{it} = y_{it} - E[y_{it}\mid x_i]
= -(u_i-E[u_i \mid  x_i]) + v_{it},
$
the between and within quantities in \eqref{eq:within_id}--\eqref{eq:between_id} are simply computed conditional on $x_i$ (here $\bar x_i = x_i$), and the identities \eqref{eq:mu2v}--\eqref{eq:mu4v} and \eqref{eq:mu2u}--\eqref{eq:mu4u} hold at each $x_i$.

\subsection{Repeated Measurements with Discrete $x$}
If $x_i$, further, takes only finitely many values, the repeated measurement setting above applies for each value. The central moments of the deviation and the random error can be estimated separately for each value $\{i : x_i = x\}$ by applying the relationships \eqref{eq:mu2v}--\eqref{eq:mu4v} and \eqref{eq:mu2u}--\eqref{eq:mu4u} conditional on $x$. For example, \citet{BenMosheGenesoveRegulation} exploit a hierarchical setting in which multiple apartment sales within the same building provide repeated measurements of the equilibrium price at a given building location.

\subsection{Panel Data with $u$ and $v$ Independent of $x$} \label{ap:homo}

A further special case of the panel-data setting in the main text arises when the distributions of $u_i$ and $v_{it}$ do not depend on $x_{it}$. In this case, $u_i$ and $v_{it}$ are identically distributed across firms and time periods, and the identities \eqref{eq:mu2v}--\eqref{eq:mu4v} and \eqref{eq:mu2u}--\eqref{eq:mu4u} hold unconditionally. We can therefore pool information across firms and time periods to estimate the (unconditional) central moments of the random error and the deviation by averaging the corresponding firm-level estimators.
{\small
	\begin{align}
		\widehat\mu_{2,v} 
		&= \left(\sum_{i=1}^n (T_i-1)\right)^{-1}\sum_{i=1}^n\sum_{t=1}^{T_i} (\hat \eps_{it}^w)^2, \label{eq:mu2vhat_pool} \qquad 
		\widehat\mu_{3,v} 
		= \left(\sum_{i=1}^n \frac{(T_i-1)(T_i-2)}{T_i}\right)^{-1}\sum_{i=1}^n\sum_{t=1}^{T_i} (\hat \eps_{it}^w)^3, \\ 
		\widehat\mu_{2,u} &= \frac{1}{n-1} \sum_{i=1}^n (\bar {\hat \eps}_i )^2 - \frac{\widehat\mu_{2,v}}{n} \sum_{i=1}^n \frac{1}{T_i}, \label{eq:mu2uhat_pool} \qquad 
		\widehat\mu_{3,u} = -\frac{n}{(n-1)(n-2)} \sum_{i=1}^n (\bar {\hat \eps}_i )^3 + \frac{\widehat\mu_{3,v}}{n} \sum_{i=1}^n \frac{1}{T_i^2}, \\ 
		\widehat\mu_{4,v} 
		&= \frac{
			\Big( \sum_{i=1}^n \frac{(T_i-1)(T_i^3-2T_i^2-3T_i+9)}{T_i^2} \Big)  \sum_{i,t}(\hat\eps_{it}^w)^4 - 6\Big(\sum_{i=1}^n \frac{(T_i-1)(2T_i-3)}{T_i^2}\Big) \sum_{i}\sum_{t<t'}(\hat\eps_{it}^w)^2(\hat\eps_{it'}^w)^2
		}{
			\Big(\sum_{i=1}^n \frac{(T_i-1)(T_i^2-3T_i+3)}{T_i^2}\Big) \Big(\sum_{i=1}^n \frac{(T_i-1)(T_i^3-2T_i^2-3T_i+9)}{T_i^2}\Big) - 3\Big(\sum_{i=1}^n \frac{(T_i-1)(2T_i-3)}{T_i^2}\Big)^2
		}, \label{eq:mu4vhat_pool}\\[8pt]
		\widehat{\mu_{2,v}^2}
		&= \frac{
			2\Big(\sum_{i=1}^n \frac{(T_i-1)(T_i^2-3T_i+3)}{T_i^2}\Big) \sum_{i}\sum_{t<t'}(\hat\eps_{it}^w)^2(\hat\eps_{it'}^w)^2 - \Big(\sum_{i=1}^n \frac{(T_i-1)(2T_i-3)}{T_i^2}\Big)  \sum_{i,t}(\hat\eps_{it}^w)^4
		}{
			\Big(\sum_{i=1}^n \frac{(T_i-1)(T_i^2-3T_i+3)}{T_i^2}\Big) \Big(\sum_{i=1}^n \frac{(T_i-1)(T_i^3-2T_i^2-3T_i+9)}{T_i^2}\Big) - 3\Big(\sum_{i=1}^n \frac{(T_i-1)(2T_i-3)}{T_i^2}\Big)^2
		}. \label{eq:mu2vsqhat_pool} \\
		\widehat\mu_{4,u}  
		&= \frac{n}{(n-1)(n^4+3n^3-26n^2+48n-24)} \Bigg[    \label{eq:mu4uhat_pool}\\
		&\quad\quad (n^3+6n^2-11n+9)\Bigg( \sum_{i}(\bar{\hat\eps}_i)^4 - 6\widehat\mu_{2,u}\widehat\mu_{2,v}\frac{(n-1)^2}{n^2}\sum_{i=1}^n\frac{1}{T_i} - \widehat\mu_{4,v}\frac{(n-1)(n^2-3n+3)}{n^3}\sum_{i=1}^n\frac{1}{T_i^3} \nonumber\\
		&\quad - 3\widehat{\mu_{2,v}^2}\left( \frac{(n-2)^2}{n^2}\sum_{i=1}^n\frac{1}{T_i^2} + \frac{2n-3}{n^3}\Big(\sum_{i=1}^n\frac{1}{T_i}\Big)^{\!2} - \frac{(n-1)(n^2-3n+3)}{n^3}\sum_{i=1}^n\frac{1}{T_i^3} \right) \Bigg) \nonumber\\
		&\quad - 3(2n-3)\Bigg( \sum_{i\neq i'}(\bar{\hat\eps}_i)^2(\bar{\hat\eps}_{i'})^2 - 2\widehat\mu_{2,u}\widehat\mu_{2,v}\frac{(n-1)\big((n-1)^2+2\big)}{n^2}\sum_{i=1}^n\frac{1}{T_i} - \widehat\mu_{4,v}\frac{(n-1)(2n-3)}{n^3}\sum_{i=1}^n\frac{1}{T_i^3} \nonumber\\
		&\quad - \widehat{\mu_{2,v}^2}\left( \frac{n^3-2n^2-3n+9}{n^3}\Big(\sum_{i=1}^n\frac{1}{T_i}\Big)^{\!2} + \frac{(n-2)(6-n)}{n^2}\sum_{i=1}^n\frac{1}{T_i^2} - \frac{3(n-1)(2n-3)}{n^3}\sum_{i=1}^n\frac{1}{T_i^3} \right) \Bigg) \Bigg]. \nonumber
	\end{align}
}These formulas are obtained by replacing the conditional central moments in \eqref{eq:mu2v}--\eqref{eq:mu4v} and \eqref{eq:mu2u}--\eqref{eq:mu4u} with their unconditional counterparts and pooling across firms and time periods.

\subsection{Cross-Section} \label{ap:cs}

Consider the model \eqref{eq:modelerror1}--\eqref{eq:modelerror2} with cross-sectional data $\{y_i,x_i\}_{i=1}^n$:
\[
y_i = g(x_i) - u_i + v_i,\qquad u_i \geq 0.
\]

In addition to Assumption \ref{asn:v} ($E[v \mid x] = 0$ and $u\indep v \mid x$), assume that $v$ has zero conditional skewness, $\mu_{3v}(x) = 0$. Then the conditional central moment equations are
\begin{align*}
	\mu_{2,y}(x) &= \mu_{2,u}(x) + \mu_{2,v}(x), \quad
	\mu_{3,y}(x) = -\mu_{3,u}(x), \quad
	\mu_{4,y}(x) = \mu_{4,u}(x) + 6\mu_{2,u}(x)\mu_{2,v}(x) + \mu_{4,v}(x).
\end{align*}
Add the absolute first and third conditional centered moments \citep{parmeter2023alternative}:
\begin{align*}
	E\left[\left|y - E[y \mid x]\right| \mid x\right]
	&= E\left[\left|\,v - \bigl(u - E[u \mid x]\bigr)\right| \mid x\right], \\
	E\left[\left|y - E[y \mid x]\right|^3 \mid x\right]
	&= E\left[\left|v - \bigl(u - E[u \mid x]\bigr)\right|^3 \mid x\right].
\end{align*}

Estimation proceeds as follows. First, nonparametrically estimate $E[y \mid x]$ and obtain residuals. Second, use these residuals to estimate the conditional moments of $y$. Third, at each input value, specify parametric distributions for $u \mid x$ and $v \mid x$, estimate their parameters by method of moments subject to the near-frontier mass constraint, and compute $\widehat E[u \mid x]$. Finally, estimate the FSF by $\widehat g(x)=\widehat E[y\mid x]+\widehat E[u\mid x].$

\section{Additional Table for Monte Carlo Simulations} \label{ap:tables_figures}

\begin{table}[!htb]
	\centering
	\footnotesize
	\setlength{\tabcolsep}{3.0pt}
	\begin{threeparttable}
		\caption{Bind Share and Absolute Bias by Region}
		\label{table:bind_absbias}
		\begin{tabular}{@{}ll cccc cccc cccc@{}}
			\toprule
			&& \multicolumn{4}{c}{Bind Share} & \multicolumn{4}{c}{Median $|\text{Bias}|$} & \multicolumn{4}{c}{Mean $|\text{Bias}|$} \\
			\cmidrule(lr){3-6} \cmidrule(lr){7-10} \cmidrule(lr){11-14}
			$m_0$ & $c$ & High-NFM & Uni-R & Uni-L & Low-NFM & High-NFM & Uni-R & Uni-L & Low-NFM & High-NFM & Uni-R & Uni-L & Low-NFM \\
			\midrule
			\multicolumn{14}{@{}l}{\textit{Panel A: ($n=25$)}} \\
			- & $\infty$ & -- & -- & -- & -- & 0.094 & 0.210 & 0.311 & 1.806 & $>\!10^{14}$ & 0.483 & 3.955 & 5.185 \\ 
			1 & $1.0$ & 0.053 & 0.174 & 0.410 & 0.731 & 0.094 & 0.377 & 1.384 & 2.371 & 0.182 & 0.440 & 1.441 & 2.259 \\ 
			1 & $0.5$ & 0.114 & 0.307 & 0.504 & 0.813 & 0.094 & 0.424 & 1.468 & 2.421 & 0.192 & 0.484 & 1.519 & 2.309 \\ 
			\midrule
			\multicolumn{14}{@{}l}{\textit{Panel B: ($n=250$)}} \\
			- & $\infty$ & -- & -- & -- & -- & 0.039 & 0.240 & 0.933 & 1.178 & 0.087 & 0.290 & 23 & $>\!10^{63}$ \\ 
			1 & $1.0$ & 0.002 & 0.021 & 0.277 & 0.560 & 0.039 & 0.241 & 1.046 & 1.779 & 0.087 & 0.294 & 1.110 & 1.774 \\ 
			1 & $0.5$ & 0.006 & 0.064 & 0.401 & 0.679 & 0.039 & 0.242 & 1.093 & 1.871 & 0.087 & 0.301 & 1.158 & 1.834 \\ 
			\midrule
			\multicolumn{14}{@{}l}{\textit{Panel C: ($n=2{,}500$)}} \\
			- & $\infty$ & -- & -- & -- & -- & 0.013 & 0.214 & 1.024 & 1.287 & 0.060 & 0.270 & 1.064 & $>\!10^{92}$ \\ 
			1 & $1.0$ & 0.000 & 0.002 & 0.185 & 0.466 & 0.013 & 0.214 & 1.026 & 1.483 & 0.060 & 0.270 & 1.073 & 1.557 \\ 
			1 & $0.5$ & 0.001 & 0.011 & 0.251 & 0.560 & 0.013 & 0.214 & 1.027 & 1.567 & 0.060 & 0.270 & 1.078 & 1.627 \\ 
			\bottomrule
		\end{tabular}
\begin{tablenotes}[para,flushleft]
	\hspace*{-0.4em} Notes:
	Regions are defined by shape parameters $(a,b)$:
	High-NFM, high near-frontier mass ($a<1$);
	Uni-R, unimodal right-skewed ($1 \le a \le b$);
	Uni-L, unimodal left-skewed ($1 \le b < a$);
	Low-NFM, low near-frontier mass ($a \ge 1,\, b<1$).
	For each grid point, statistics are first computed across the 150 Monte Carlo replications.
	The reported entries then aggregate these grid-point statistics across all grid points in the region.
	Bind Share is the mean across grid points of the fraction of replications in which the near-frontier mass constraint binds.
	Median $|\text{Bias}|$ is the median across grid points of the absolute bias of $\widehat{E}[u]$.
	Mean $|\text{Bias}|$ is the mean across grid points of the absolute bias of $\widehat{E}[u]$.
\end{tablenotes}
	\end{threeparttable}
\end{table}

\end{document}